\documentclass[preprint,12pt,envcountsame]{elsarticle} %
\biboptions{sort&compress}

\usepackage[T1]{fontenc}

\usepackage{amsthm}
\usepackage{calc}
\newtheorem{theorem}{Theorem}[section]
\newtheorem{lemma}[theorem]{Lemma}

\newtheorem{proposition}[theorem]{Proposition}

\newtheorem{remark}[theorem]{Remark}

\newtheorem{definition}[theorem]{Definition}
\newtheorem{example}[theorem]{Example}

\title{Qualitative Reachability in Stochastic BPA Games}

\author[fimu]{Tom\'{a}\v{s} Br\'{a}zdil}
\ead{brazdil@fi.muni.cz}
\author[fimu]{V\'{a}clav Bro\v{z}ek}
\ead{brozek@fi.muni.cz}
\author[fimu]{Anton\'{\i}n Ku\v{c}era}
\ead{kucera@fi.muni.cz}
\author[fimu]{Jan Obdr\v{z}\'{a}lek}
\ead{obdrzalek@fi.muni.cz}
\address[fimu]{Faculty of Informatics, Masaryk University,\\
Botanick\'a 68a, 60200 Brno\\
Czech Republic
}

\usepackage{amssymb}
\usepackage{amsmath}
\usepackage{mathrsfs}

\usepackage{wasysym}
\renewcommand{\bigcirc}{\ocircle}

\usepackage{tikz}
\usepackage{pgffor}
\usetikzlibrary{arrows,shapes}
\tikzstyle{max}=[thick,draw,minimum size=1.2em,inner sep=.2ex]
\tikzstyle{min}=[diamond,thick,draw,minimum size=1.2em,inner sep=.2ex]
\tikzstyle{ran}=[ellipse,thick,draw,minimum size=1.2em,inner sep=.2ex]
\tikzstyle{transition}=[thick,draw,->,>=stealth]
\tikzstyle{loop left}=[transition, to path={.. controls +(150:.5) 
    and +(210:.5) .. (\tikztotarget) \tikztonodes}]
\tikzstyle{loop right}=[transition, to path={.. controls +(30:.5) 
    and +(330:.5) .. (\tikztotarget) \tikztonodes}]
\tikzstyle{loop above}=[transition, to path={.. controls +(60:.5) 
    and +(120:.5) .. (\tikztotarget) \tikztonodes}]
\tikzstyle{loop below}=[transition, to path={.. controls +(230:.5) 
    and +(300:.5) .. (\tikztotarget) \tikztonodes}]

\usepackage[ruled,vlined,linesnumbered,lined]{algorithm2e}

\newcommand{\Att}{\mathit{Att}}
\newcommand{\calF}{\mathcal{F}}

\newcommand{\calO}{\mathcal{O}}
\newcommand{\calP}{\mathcal{P}}

\newcommand{\A}{\mathscr{A}}
\newcommand{\calA}{\mathcal{A}}
\newcommand{\B}{\mathscr{B}}
\newcommand{\C}{\mathscr{C}}
\newcommand{\scrD}{\mathscr{D}}
\newcommand{\D}{\mathcal{D}}
\newcommand{\F}{\mathcal{F}}
\newcommand{\G}{\mathcal{G}}
\newcommand{\U}{\mathcal{U}}
\newcommand{\V}{\mathcal{V}}
\newcommand{\W}{\mathcal{W}}
\newcommand{\M}{\mathcal{M}}

\newcommand{\Prob}{\mathit{Prob}}
\newcommand{\Nset}{\mathbb{N}}

\newcommand{\Qset}{\mathbb{Q}}
\newcommand{\Rset}{\mathbb{R}}
\newcommand{\R}{\mathit{Reach}}
\newcommand{\fpath}{\mathit{FPath}}
\newcommand{\run}{\mathit{Run}}
\newcommand{\price}{\mathit{price}}
\newcommand{\eps}{\varepsilon}

\newcommand{\wsetB}[3]{[{#1}]_{\Box}^{{#2}#3}}%
\newcommand{\wsetD}[3]{[{#1}]_{\Diamond}^{{#2}#3}}%
\newcommand{\NP}{\textbf{NP}}
\newcommand{\PTIME}{\textbf{P}}
\newcommand{\val}{\mathit{val}}
\newcommand{\len}[1]{|#1|}
\newcommand{\coNP}{\textbf{co-NP}}
\newcommand{\dist}{\mathcal{D}}

\newcommand{\pp}[1]{\widetilde{#1}}

\newcommand{\initdelta}{\bar{\Delta}}

\renewcommand{\succ}{\rhd}
\newcommand{\notsucc}{\mathop{\not\!\rhd}}

\newcommand{\coloneqq}{\mathrel{\mathop:}=}

\newcommand{\tran}[1]{{}\mathchoice%
    {\stackrel{#1}{\longrightarrow}}
    {\mathop {\smash\longrightarrow}\limits^{\vrule width 0pt height 0pt 
                                                depth 4pt\smash{#1}}}
    {\stackrel{#1}{\longrightarrow}}
    {\stackrel{#1}{\longrightarrow}}
{}}
\newcommand{\btran}[2][]{{}\mathchoice%
    {\stackrel{#2}{\hookrightarrow_{#1}}}
    {\mathop {\smash{\hookrightarrow_{#1}}}\limits^{\vrule width 0pt height 0pt 
                                                depth 4pt\smash{#2}}}
    {\stackrel{#2}{\hookrightarrow_{#1}}}
    {\stackrel{#2}{\hookrightarrow_{#1}}}
{}}
\newcommand{\bbtran}[2][]{{}\mathchoice%
    {\stackrel{#2}{\leadsto_{#1}}}
    {\mathop {\smash{\leadsto_{#1}}}\limits^{\vrule width 0pt height 0pt 
                                                depth 4pt\smash{#2}}}
    {\stackrel{#2}{\leadsto_{#1}}}
    {\stackrel{#2}{\leadsto_{#1}}}
{}}

\newcommand{\gtran}[2][]{{}\mathchoice%
    {\stackrel{#2}{\mapsto_{#1}}}
    {\mathop {\smash{\mapsto_{#1}}}\limits^{\vrule width 0pt height 0pt 
                                                depth 4pt\smash{#2}}}
    {\stackrel{#2}{\mapsto_{#1}}}
    {\stackrel{#2}{\mapsto_{#1}}}
{}}

\newcommand{\qedv}[1]{\ifmmode\squareforqed\else{\unskip\nobreak\hfil
\penalty50\hskip1em\null\nobreak\hfil(#1)\squareforqed
\parfillskip=0pt\finalhyphendemerits=0\endgraf}\fi}

\begin{document}

\begin{abstract}
We consider a class of infinite-state stochastic games generated by
stateless pushdown automata (or, equivalently, 1-exit recursive state
machines), where the winning objective is specified by a regular set of
target configurations and a qualitative probability constraint `${>}0$'
 or `${=}1$'. The goal of one player is to maximize the
probability of reaching the target set so that the constraint is satisfied,
while the other player aims at the opposite.  We show that the winner in
such games can be determined in $\PTIME$ for the `${>}0$' constraint,
and in $\NP \cap \coNP$ for the `${=}1$' constraint. Further, we prove that the
winning regions for both players are regular, and we design algorithms which
compute the associated finite-state automata.  Finally, we show that winning
strategies can be synthesized effectively.
\end{abstract}

\maketitle

\section{Introduction}
\label{sec-intro}

Stochastic games are a formal model for discrete systems where the
behavior in each state is either controllable, adversarial, or
stochastic. Formally, a stochastic game is a directed graph $G$ 
with a denumerable set of
vertices $V$ which is split into three disjoint subsets $V_\Box$,
$V_\Diamond$, and $V_\bigcirc$.  For every $v \in V_\bigcirc$, there
is a fixed probability distribution over the outgoing edges
of~$v$. We also require that the set of outgoing edges of
every vertex is nonempty.
The game is initiated by putting a token on some
vertex. The token is then moved from vertex to vertex by two players,
$\Box$ and $\Diamond$, who choose the next move in the vertices of
$V_\Box$ and $V_\Diamond$, respectively.  In the vertices of
$V_\bigcirc$, the outgoing edges are chosen according to the
associated fixed probability distribution. 
A \emph{quantitative winning objective} is specified by some Borel 
set $W$ of infinite paths in $G$ and a probability constraint 
${\succ}\varrho$, where ${\succ} \in \{{>},{\geq}\}$ 
is a comparison and $\varrho \in [0,1]$.
An important subclass of quantitative winning objectives are
\emph{qualitative winning objectives} where the constant $\varrho$ 
must be either $0$ or $1$. The goal of player~$\Box$ is to 
maximize the probability of all runs that stay in $W$ so that it
is $\succ$-related to $\varrho$, while player~$\Diamond$ aims at the 
opposite. A \emph{strategy}
specifies how a player should play. In general, a strategy may or may
not depend on the history of a play (we say that a strategy is
\emph{history-dependent (H)} or \emph{memoryless (M)}), and the
edges may be chosen deterministically or randomly
(\emph{deterministic (D)} and \emph{randomized (R)} strategies). In
the case of randomized strategies, a player chooses a probability
distribution on the set of outgoing edges.  Note that
deterministic strategies can be seen as restricted randomized
strategies, where one of the outgoing edges has probability~$1$.
Each pair of strategies $(\sigma,\pi)$ for players $\Box$ and
$\Diamond$ determines a \emph{play}, i.e., a unique Markov chain 
obtained from $G$ by applying the strategies $\sigma$ and $\pi$
in the natural way. The \emph{outcome} of a play initiated in
$v$ is the probability of all runs initiated in $v$ that are contained
in the  set $W$ (this probability is denoted by $\calP_v^{\sigma,\pi}(W)$). 
We say that a play is
$({\succ}\varrho)$-won by player~$\Box$ if its outcome is $\succ$-related
to~$\varrho$;
otherwise, the play is $({\notsucc}\varrho)$-won by player~$\Diamond$. 
A strategy $\sigma$ of player~$\Box$ is \emph{$({\succ}\varrho)$-winning}
if for every strategy $\pi$ of player~$\Diamond$, the  
corresponding play is $({\succ}\varrho)$-won by player~$\Box$.
Similarly, a strategy $\pi$ of player~$\Diamond$ is
\emph{$({\notsucc}\varrho)$-winning} if for every strategy
$\sigma$ of player~$\Box$, the  
corresponding play is $({\notsucc}\varrho)$-won by player~$\Diamond$.
A natural question is whether the game is \emph{determined}, i.e.{}, 
for every choice of ${\succ}$ and $\varrho$, either player~$\Box$ has
a $({\succ}\varrho)$-winning strategy, or player~$\Diamond$
has a $({\notsucc}\varrho)$-winning strategy. The answer is somewhat 
subtle. A celebrated result of 
Martin \cite{Martin:Blackwell-determinacy} (see also 
\cite{MS:stochastic-games-determinacy}) implies that stochastic games 
with Borel winning conditions are \emph{weakly determined}, i.e.,
each vertex $v$ has a \emph{value} given by
\begin{equation}
\label{eq-determinacy}
  \val(v) \quad = \quad \sup_\sigma \inf_\pi \calP_v^{\sigma,\pi}(W) 
  \quad = \quad \inf_\pi \sup_\sigma \calP_v^{\sigma,\pi}(W)
\end{equation}
Here $\sigma$ and $\pi$ range over the sets of all strategies for
player~$\Box$ and player~$\Diamond$, respectively. From this we can
immediately deduce the following: 
\begin{itemize}
\item If both players have \emph{optimal} strategies that guarantee
  the outcome $\val(v)$ or better against every strategy of the
  opponent (for example, this holds for \emph{finite-state} stochastic
  games and the ``usual'' classes of quantitative/qualitative Borel
  objectives), then the game is determined for every choice of 
  $\succ\varrho$.
\item Although optimal strategies are not guaranteed to exists in general,
  Equation~\ref{eq-determinacy} implies the existence
  of $\varepsilon$-optimal strategies  (see 
  Definition~\ref{def-optimal-strategy}) for every $\varepsilon > 0$.
  Hence, the game is determined for every choice of ${\succ}\varrho$
  where $\varrho \neq \val(v)$.  
\end{itemize}
The only problematic case is the situation when optimal strategies do not 
exist and $\varrho = \val(v)$. The example given in 
Figure~\ref{fig-not-determined} at page~\pageref{fig-not-determined} 
witnesses that such games are generally \emph{not} determined, even 
for reachability objectives. On the other hand, we show that 
\emph{finitely-branching} games (such as BPA games considered
in this paper) with reachability objectives \emph{are}
determined, although an optimal strategy for player~$\Box$
in a finitely-branching game 
does \emph{not} necessarily exist. The determinacy question for 
finitely-branching games and other classes of (Borel) winning 
objectives is left open.

Algorithmic issues for stochastic games with quantitative/qualitative 
winning objectives have been studied mainly for finite-state stochastic 
games. A lot of attention has been devoted to quantitative 
\emph{reachability objectives}, including the special case
when $\varrho = \frac{1}{2}$. 
The problem whether player~$\Box$ has a $({>}\frac{1}{2}$)-winning
strategy is known to be in $\NP \cap \coNP$, but its membership
to $\PTIME$ is a long-standing open problems in
algorithmic game theory 
\cite{Condon:simple-stochastic-games-IC,Vieille:stochastic-games-HGT}.
Later, more complicated qualitative/quantitative $\omega$-regular
winning objectives (such as B\"{u}chi, co-B\"{u}chi, Rabin, Street, 
Muller etc.) were considered, and the complexity of the corresponding
decision problems was analyzed. We refer to 
\cite{AM:omega-regular-games-JCSS,CHAH:Rabin-Street-Games,%
CHJH:quant-parity,CHJH:simple-stochastic-parity,%
Walukiewicz:games-background,Thomas:games-verification} for more details. 
As for infinite-state stochastic games, the attention has so far been
focused on stochastic games induced by lossy channel systems
\cite{AHAMS:Stochastic-games-lossy,BBS:MDP-LMC-omega-regular-Trans}
and by pushdown automata (or, equivalently, recursive state machines)
\cite{EY:RMDP-efficient,EY:RSCG,EY:RMC-RMDP,EWY:RSG-Positive-Rewards,%
BBFK:BPA-games-reachability-IC}. In the next paragraphs, 
we discuss the latter model in greater
detail because these results are closely related to the results presented in 
this paper.

A \emph{pushdown automaton (PDA)} (see, e.g., \cite{HU:book})
is equipped with a finite control unit and an unbounded
stack. The dynamics is specified by a finite set of rules of
the form $pX \btran{} q\alpha$, where $p,q$ are control states,
$X$ is a stack symbol, and $\alpha$ is a (possibly empty) sequence
of stack symbols. A rule of the form $pX \btran{} q\alpha$ is applicable
to every configuration of the form $pX\beta$ and produces the
configuration $q\alpha\beta$. If there are several rules with the
same left-hand side, one of them must be chosen, and the choice
is made by player~$\Box$, player~$\Diamond$, or it is randomized. 
Technically, the set of all left-hand sides (i.e., pairs of the form $pX$) 
is split into three disjoint subsets $H_\Box$, $H_\Diamond$, and $H_\bigcirc$,
and for all $pX \in H_\bigcirc$ there is
a fixed probability distribution over the set of all rules of the form
$pX \btran{} q\alpha$. Thus, each PDA induces the associated
infinite-state stochastic game where the vertices are PDA 
configurations and the edges are determined in the natural way. 
An important subclass of PDA is obtained by restricting the number
of control states to~$1$. Such PDA are also known as \emph{stateless}
PDA or (mainly in concurrency theory) as BPA. PDA and BPA correspond to
\emph{recursive state machines (RSM)} and \emph{1-exit RSM} respectively, in the sense
that their descriptive powers are equivalent, and there are effective
linear-time translations between the corresponding models.

In \cite{EY:RMC-RMDP}, the quantitative and qualitative 
\emph{termination objective} for PDA and BPA stochastic games
is examined (a terminating run is a run which hits a
configuration with the empty stack; hence, termination 
is a special form of reachability). For BPA, it is 
shown that the vector of optimal values $(\val(X), X \in \Gamma)$, 
where $\Gamma$ is the stack alphabet, forms the least solution
of an effectively constructible system of min-max equations.
Moreover, both players have \emph{optimal} MD strategies which depend
only on the topmost stack symbol of a given configuration
(such strategies are called SMD, meaning Stackless MD). Hence, 
stochastic BPA games with 
quantitative/qualitative termination objectives are determined.
Since the least solution of the constructed equational system
can be encoded in first order theory of the reals, the existence
of a $({\succ}\varrho)$-winning strategy for player~$\Box$
can be decided in polynomial space. In the same paper
\cite{EY:RMC-RMDP}, the $\Sigma_2^P \cap \Pi_2^P$ upper complexity
bound for the subclass of qualitative termination 
objectives is established. As for PDA games, it is shown that for 
every fixed $\varepsilon>0$,
the problem to distinguish whether the optimal value $\val(pX)$
is equal to~$1$ or less than $\varepsilon$, is undecidable.
The $\Sigma_2^P \cap \Pi_2^P$ upper bound for stochastic BPA games 
with qualitative termination objectives is improved to
$\NP \cap \coNP$ in \cite{EY:RMDP-efficient}. In the same paper, it is also 
shown that the quantitative reachability problem for finite-state stochastic
games (see above) is efficiently reducible to the qualitative termination 
problem for stochastic BPA games. Hence, the $\NP \cap \coNP$ upper bound
cannot be further improved without a major breakthrough in algorithmic game
theory. In the special case of stochastic BPA games where
$H_\Diamond = \emptyset$ or $H_\Box = \emptyset$, the qualitative
termination problem  is shown to be
in $\PTIME$ (observe that if $H_\Diamond = \emptyset$ or $H_\Box = \emptyset$,
then a given BPA induces an infinite-state Markov decision process and 
the goal of the only player is to maximize or minimize the termination
probability, respectively). The results for Markov decision processes
induced by BPA are generalized to (arbitrary) qualitative 
\emph{reachability objectives} in \cite{BBFK:BPA-games-reachability-IC},
retaining the $\PTIME$ upper complexity bound. In the same paper,
it is also noted that the properties of reachability
objectives are quite different from the ones of termination
(in particular, there is no apparent way how to express
the vector of optimal values as a solution of some recursive equational 
system, and the SMD determinacy result (see above) does not hold either).

\textbf{Our contribution:} In this paper, we continue the study initiated
in \cite{EY:RMDP-efficient,EY:RSCG,EY:RMC-RMDP,EWY:RSG-Positive-Rewards,%
BBFK:BPA-games-reachability-IC} and solve the qualitative reachability
problem for unrestricted stochastic BPA games. Thus, we obtain 
a substantial generalization of the previous results. 

We start by resolving the determinacy issue in 
Section~\ref{sec-determinacy}. We observe that general stochastic
games with reachability objectives are \emph{not} determined,
and we also show that \emph{finitely branching} stochastic games
(such as BPA stochastic games) with 
quantitative/qualitative reachability objectives \emph{are} determined, i.e.,
in every vertex, either player~$\Box$ has 
a $({\succ}\varrho)$-winning strategy, or player~$\Diamond$ has 
a $({\notsucc}\varrho)$-winning strategy. This is a consequence
of several observations that are specific to reachability objectives and
perhaps interesting on their own. 

The main results of our paper, presented in 
Sections~\ref{sec:zero},~\ref{sec:one}, and~\ref{sec:proofs}
concern stochastic BPA games with  
qualitative reachability objectives. In the context
of BPA, a reachability objective is specified by a \emph{regular}
set $T$ of target configurations. 
We show that the problem of determining the winner in stochastic BPA games
with qualitative reachability objectives is in $\PTIME$ for the
`${\succ}0$' constraint, and in $\NP \cap \coNP$ for the `${\succ}1$'
constraint. Here we rely on the previously discussed results about 
qualitative termination \cite{EY:RMDP-efficient} and use the corresponding
algorithms as ``black-box procedures'' at appropriate places. 
We also rely on observations presented in  
\cite{BBFK:BPA-games-reachability-IC} which were used to solve the
simpler case with only one player. However, the full (two-player)
case brings completely new complications that need to be tackled by
new methods and ideas. Many ``natural'' hypotheses turned out to be
incorrect (some of the interesting cases are documented by
examples in Section~\ref{sec-BPA-games}). 
We also show that for each $\varrho \in \{0,1\}$, the sets of all 
configurations where player~$\Box$ (or player~$\Diamond$) has a 
\mbox{$({\succ}\varrho)$-winning} (or \mbox{$({\notsucc}\varrho)$-winning})
strategy is effectively regular, and the
corresponding finite-state automaton is effectively constructible
by a deterministic polynomial-time algorithm (for the `${\succ}1$'
constraint, the algorithm needs $\NP \cap \coNP$ oracle). 
Finally, we also give algorithms which
\emph{compute} winning strategies
if they exist. These strategies are memoryless, and they
are also \emph{effectively regular} in the sense that their 
functionality is effectively expressible by finite-state automata
(see Definition~\ref{def-regular-strategy}). Hence, 
winning strategies in stochastic BPA games with qualitative
reachability objectives can be effectively implemented.

For the sake of readability, some of the more involved (and long) proofs
of Section~\ref{sec:one} have been postponed to
Section~\ref{sec:proofs}. In the main body of the paper, we try to
sketch the key ideas and provide some intuition behind the presented
technical constructions.

\section{Basic Definitions}
\label{sec-defs}

In this paper, the sets of all positive integers, non-negative
integers, rational numbers, real numbers, and non-negative real
numbers are denoted by $\Nset$, $\Nset_0$, $\Qset$, $\Rset$, and
$\Rset^{\geq 0}$, respectively. 
For every finite or countably infinite set $S$, the symbol $S^*$
denotes the set of all finite words over $S$.
The length of a given word $u$ is denoted by $|u|$, and the individual
letters in $u$ are denoted by $u(0),\cdots,u(|u|-1)$.
The empty word is denoted by $\varepsilon$, and we set
$|\varepsilon| = 0$.  We also use $S^+$ to denote the
set $S^* \smallsetminus \{\varepsilon\}$. 
For every finite or countably infinite set $M$, a binary relation 
${\to} \subseteq M \times M$ is \emph{total} if for every
$m \in M$ there is some $n \in M$ such that $m \to n$.
A \emph{path} in $\M = (M,{\to})$ is a finite or
infinite sequence $w = m_0,m_1,\ldots$ such that 
$m_i \to m_{i+1}$ for every $i$. The \emph{length} of a finite
path $w = m_0,\ldots,m_i$, denoted by $\len{w}$, is $i+1$.
We also use $w(i)$ to denote the element
$m_i$ of $w$, and $w_i$ to denote the path $m_i,m_{i+1},\ldots$ 
(by writing $w(i) = m$ or $w_i$ we implicitly impose the condition
that $\len{w} \geq i{+}1$). A given $n \in M$ is \emph{reachable} from  
a given $m \in M$, written $m \to^* n$, if there is a finite path 
from $m$ to $n$. A \emph{run} is an infinite path. The sets of all
finite paths and all runs in $\M$ are denoted by 
$\fpath(\M)$ and $\run(\M)$,
respectively.  Similarly, the sets of all finite paths and runs that
start in a given $m \in M$ are denoted by $\fpath(\M,m)$ and $\run(\M,m)$,
respectively.  

Now we recall basic notions of probability theory. 
Let $A$ be a finite or countably infinite set. A 
\emph{probability distribution}
on $A$ is a function $f : A \rightarrow \Rset^{\geq 0}$ such that
\mbox{$\sum_{a \in A} f(a) = 1$}. A distribution $f$ is \emph{rational}
if $f(a) \in \Qset$ for every $a \in A$,
\emph{positive} if $f(a) > 0$ for every $a \in A$, \emph{Dirac}
if $f(a) = 1$ for some $a \in A$, and \emph{uniform} if $A$ is
finite and $f(a) = \frac{1}{|A|}$ for every $a \in A$.
The set of all distributions on $A$ is denoted by $\D(A)$. 

A \emph{$\sigma$-field} over a set $X$ is a set $\calF \subseteq 2^X$
that includes $X$ and is closed under complement and countable union. 
A \emph{measurable space} is a pair $(X,\calF)$ where $X$ is a
set called \emph{sample space} and $\calF$ is a $\sigma$-field over $X$. 
A \emph{probability measure} over a measurable space
$(X,\calF)$ is a function $\calP : \calF \rightarrow \Rset^{\geq 0}$
such that, for each countable collection $\{X_i\}_{i\in I}$ of pairwise
disjoint elements of $\calF$, $\calP(\bigcup_{i\in I} X_i) = 
\sum_{i\in I} \calP(X_i)$, and moreover $\calP(X)=1$. A 
\emph{probability space} is a triple $(X,\calF,\calP)$ where 
$(X,\calF)$ is a measurable space and $\calP$
is a probability measure over $(X,\calF)$. 

\begin{definition}
A \emph{Markov chain} is a triple \mbox{$\M = (M,\tran{},\Prob)$} 
where $M$ is a finite or countably infinite set
of \emph{states}, ${\tran{}} \subseteq M \times M$ is a total 
\emph{transition relation}, and $\Prob$ is a function which to each
$s \in M$ assigns a positive probability distribution over the set
of its outgoing transitions. 
\end{definition}

In the rest of this paper, we write $s \tran{x} t$ whenever
$s \tran{} t$ and $\Prob((s,t)) = x$. 
Each $w \in \fpath(\M)$ determines a \emph{basic cylinder} $\run(\M,w)$
which consists of all runs that start with $w$.  To every $s \in M$
we associate the probability space $(\run(\M,s),\calF,\calP)$ where 
$\calF$ is the $\sigma$-field generated by all basic cylinders
$\run(\M,w)$ where $w$ starts with $s$, and 
\mbox{$\calP: \calF \rightarrow \Rset^{\geq 0}$}
 is the unique probability measure such that
$\calP(\run(\M,w)) = \Pi_{i=0}^{m-1} x_i$ where $w = s_0,\cdots,s_m$ and 
$s_i \tran{x_i} s_{i+1}$ for every $0 \leq i < m$ (if $m=0$, we put
$\calP(\run(\M,w)) = 1$). 

\begin{definition}
A \emph{stochastic game} is a tuple 
$G = (V,\gtran{},(V_{\Box},V_\Diamond,V_{\bigcirc}),\Prob)$ where
$V$ is a finite or countably infinite set of \emph{vertices}, ${\gtran{}}
\subseteq V \times V$ is a total \emph{edge relation},
$(V_{\Box},V_\Diamond,V_{\bigcirc})$ is a partition of $V$, and $\Prob$ is a
\emph{probability assignment} which to each $v \in V_{\bigcirc}$
assigns a positive probability distribution on the set of its outgoing
edges. We say that $G$ is
\emph{finitely branching} if for each $v \in V$ there are only
finitely many $u \in V$ such that $v \gtran{} u$. 
\end{definition}

A stochastic game $G$ is played by two players, $\Box$ and $\Diamond$,
who select the moves in the vertices of $V_{\Box}$ and $V_{\Diamond}$,
respectively. Let $\odot \in \{\Box,\Diamond\}$.
A \emph{strategy} for player~$\odot$ in $G$ is a 
function which to each $wv \in V^*V_{\odot}$ assigns a probability 
distribution on the set of outgoing edges of~$v$. The sets of all 
strategies for player~$\Box$ and player~$\Diamond$ in $G$ are denoted by 
$\Sigma_G$ and $\Pi_G$ (or just by $\Sigma$ and $\Pi$ if $G$ is understood), 
respectively. We say that a strategy $\tau$
is \emph{memoryless (M)} if $\tau(wv)$ depends just on the last
vertex $v$, and \emph{deterministic (D)} if $\tau(wv)$ is a Dirac
distribution for all $wv$. Strategies that are not
necessarily memoryless are called \emph{history-dependent (H)}, and
strategies that are not necessarily deterministic are called
\emph{randomized (R)}.
Thus, we  define the following four classes of strategies:
MD, MR, HD, and HR, where $\mbox{MD} \subseteq \mbox{HD} \subseteq
\mbox{HR}$ and $\mbox{MD} \subseteq \mbox{MR} \subseteq \mbox{HR}$,
but MR and HD are incomparable. 

Each pair of strategies $(\sigma,\pi) \in \Sigma \times \Pi$ 
determines a unique \emph{play} of the game $G$, which is a Markov chain 
$G(\sigma,\pi)$ where
$V^+$ is the set of states, and $wu \tran{x} wuu'$  iff $u \gtran{} u'$
and one of the following conditions holds:
\begin{itemize}
\item $u \in V_{\Box}$ and $\sigma(wu)$ assigns $x$ to 
  $u \gtran{} u'$, where $x > 0$;
\item $u \in V_{\Diamond}$ and $\pi(wu)$ assigns $x$ to 
  $u \gtran{} u'$, where $x > 0$; 
\item $u \in V_{\bigcirc}$ and $u \gtran{x} u'$.
\end{itemize}
Let $T \subseteq V$ be a set of \emph{target} vertices. For each pair 
of strategies $(\sigma,\pi) \in \Sigma \times \Pi$
and every $v \in V$, let $\calP_v^{\sigma,\pi}(\R(T,G))$ be the probability
of all $w \in \run(G(\sigma,\pi),v)$ such that $w$ visits some $u \in T$
(technically, this means that $w(i) \in V^*T$ for some $i \in \Nset_0$).
We write $\calP_v^{\sigma,\pi}(\R(T))$ instead of
$\calP_v^{\sigma,\pi}(\R(T,G))$ if $G$ is understood.

We say that a given $v \in V$ \emph{has a value} in $G$ if
$\sup_{\sigma \in \Sigma} \inf_{\pi \in \Pi} \calP_v^{\sigma,\pi}(\R(T)) 
 = \inf_{\pi \in \Pi} \sup_{\sigma \in \Sigma} \calP_v^{\sigma,\pi}(\R(T))$.
If $v$ has a value, then $\val(v,G)$ denotes the \emph{value of $v$} 
defined by this equality (we write just $\val(v)$ instead of 
$\val(v,G)$ if $G$ is understood).
Since the set of all runs that visit a vertex of $T$ is obviously Borel,
we can apply the powerful result of Martin 
\cite{Martin:Blackwell-determinacy} (see also Theorem~\ref{thm-value})
and conclude that \emph{every} \mbox{$v \in V$} has a value.

\begin{definition}
\label{def-optimal-strategy}
Let $\varepsilon \geq 0$ and $v \in V$. We say that 
\begin{itemize}
\item $\sigma \in \Sigma$ is \emph{$\varepsilon$-optimal} (or 
\emph{$\varepsilon$-optimal maximizing}) in $v$ if 
  $\calP_v^{\sigma,\pi}(\R(T)) \geq \val(v)-\varepsilon$ for all $\pi \in \Pi$;
\item $\pi \in \Pi$ is \emph{$\varepsilon$-optimal} (or 
  \emph{$\varepsilon$-optimal minimizing}) in $v$
  if $\calP_v^{\sigma,\pi}(\R(T)) \leq \val(v)+\varepsilon$ 
  for all $\sigma \in \Sigma$.
\end{itemize}
A $0$-optimal strategy is called \emph{optimal}.
A \emph{(quantitative) reachability objective} is a pair $(T,{\succ}\varrho)$
where $T \subseteq V$ and ${\succ}\varrho$ is a probability
constraint, i.e., ${\succ} \in \{{>},{\geq}\}$ and $\varrho \in [0,1]$. 
If $\varrho \in \{0,1\}$, then the objective is \emph{qualitative}.
We say that
\begin{itemize}
\item $\sigma \in \Sigma$ is \emph{$(T,{\succ}\varrho)$-winning} in $v$ 
  if $\calP_v^{\sigma,\pi}(\R(T)) \succ \varrho$ for all $\pi \in \Pi$;
\item $\pi \in \Pi$ is \emph{$(T,{\notsucc}\varrho)$-winning} in $v$
  if $\calP_v^{\sigma,\pi}(\R(T)) \notsucc \varrho$ for all $\sigma \in \Sigma$.
\end{itemize} 
The \emph{$(T,{\succ}\varrho)$-winning region} of player~$\Box$, denoted
by $\wsetB{T}{\succ}{\varrho}$, is the set of all $v \in V$ such that
player~$\Box$ has a $(T,{\succ}\varrho)$-winning strategy in~$v$. Similarly,
the \emph{$(T,{\notsucc}\varrho)$-winning region} of player~$\Diamond$, 
denoted by $\wsetD{T}{\notsucc}{\varrho}$, consists of all $v \in V$ such that
player~$\Diamond$ has a $(T,{\notsucc}\varrho)$-winning strategy in~$v$.

When writing probability constraints, we usually use ${<}1$, ${=}1$, and ${=}0$ 
instead of ${\not\geq} 1$,  ${\geq} 1$, and ${\not >}0$, respectively. 
\end{definition}

\section{The Determinacy of Stochastic Games with 
  Reachability Objectives}
\label{sec-determinacy}

In this section we show that finitely-branching stochastic games 
with quantitative/qualitative reachability objectives are \emph{determined}
in the sense that for every quantitative reachability objective
$(T,{\succ}\varrho)$, each vertex of the game belongs either to 
$\wsetB{T}{\succ}{\varrho}$ or to $\wsetD{T}{\notsucc}{\varrho}$
(see Definition~\ref{def-optimal-strategy}).
Let us note that
this result cannot be extended to general (infinitely-branching)
stochastic games. A counterexample is given in 
Figure~\ref{fig-not-determined}, where $T = \{t\}$ is the set of
target vertices. Observe that $\val(s) = 0$, $\val(u) = 1$,
and $\val(v) = 1/2$. It is easy to check that none of the 
two players has an optimal strategy in the vertices $v$, $u$, and $s$. 
Now suppose that player~$\Box$ has a $(T,{\succ}\frac{1}{2})$-winning 
strategy $\hat{\sigma}$ in $v$. Obviously, there is some fixed 
$\varepsilon >0$ such that for every $\pi \in \Pi$ we have that
$\calP_u^{\hat{\sigma},\pi}(\R(T)) = 1 - \varepsilon$. 
Further, player~$\Diamond$ has a 
strategy $\hat{\pi}$ which is $\frac{\varepsilon}{2}$-optimal in every
vertex. Hence, $\calP_v^{\hat{\sigma},\hat{\pi}}(\R(T)) < \frac{1}{2}$, which
is a contradiction. Similarly, one can show that there is 
no $(T,{\notsucc}\frac{1}{2})$-winning strategy for player~$\Diamond$
in~$v$.

\begin{figure}[t]
\centering
\newcounter{zzz}\newcounter{jj}\newcounter{jjj}
\begin{tikzpicture}[x=1.45cm,y=1.45cm,font=\footnotesize]
\foreach \x in {0,1,2,3,4,5,6}
  {\foreach \y in {-1,-2,-3,-4} 
      {\node (v\x\y) at (\x,\y) [ran] {};}
   \node (v\x 0) at (\x,0) [max] {};
  }
\foreach \y in {0,-1,-2,-3,-4}
  {\node (v7\y) at (7,\y) [ran,draw=none] {};}
\node at (0,-2) [ran,fill=lightgray] {$t$};
\node at (0,0)  [max,draw=none] {$u$};
\node (dd) at (0,-5) [min] {$s$};
\node (v) at (-1,-2.5) [ran] {$v$};
\setcounter{jj}{1}
\foreach \i in {0,1,2,3,4,5,6}
  {\setcounter{zzz}{\i}\addtocounter{zzz}{1}
   \setcounter{jj}{2*\arabic{jj}}
   \setcounter{jjj}{\arabic{jj}-1}
   \draw [transition] (v\i 0) to (v\arabic{zzz}0);
   \draw [transition] (v\i -4) to node[left] 
      {$\frac{\arabic{jjj}}{\arabic{jj}}$} (v\i -3);
   \draw [transition,rounded corners] 
      (v\i -4) -- +(.4,.35)  -- 
      +(.4,1.65) node[right] {$\frac{1}{\arabic{jj}}$} -- (v\i -2);
  }
\foreach \i in {1,2,3,4,5,6,7}
  {\setcounter{zzz}{\i}\addtocounter{zzz}{-1}
   \draw [transition] (v\i -1) to node[above] {$\frac{1}{2}$} (v\arabic{zzz}-1);
   \draw [transition] (v\i -2) to  (v\arabic{zzz}-2);
   \draw [transition] (v\i -3) to  (v\arabic{zzz}-3);
  }
\foreach \i in {1,2,3,4,5,6}
  {\draw [transition] (v\i -1) to node[left] {$\frac{1}{2}$} (v\i-2);
   \draw [transition] (v\i 0)  to  (v\i-1);
  } 
\draw [transition] (v00) to (v0-1);
\draw [transition,loop left] (v0-1) to (v0-1);
\draw [transition,loop left] (v0-2) to (v0-2);
\foreach \i in {1,2,3,4,5,6,7}
  {\draw [transition,rounded corners] (dd) -- +(\i,0) -- (v\i-4);}
\draw [transition] (dd) to (v0-4);
\foreach \i in {0,-1,-2,-3,-4}
  {\draw [thick,dotted] (v7\i) -- +(1,0);}
\draw [transition,rounded corners] (v) -- node[left] {$\frac{1}{2}$} +(0,2.5) 
   -- (v00);
\draw [transition,rounded corners] (v) -- node[left] {$\frac{1}{2}$} +(0,-2.5) 
   -- (dd);
\draw [transition,rounded corners] (v0-3) -- +(-.6,0) -- +(-.6,-1.5) -- (dd);
\end{tikzpicture}
\caption{A game which is not determined.}
\label{fig-not-determined}
\end{figure}

For the rest of this section, let us fix a 
game \mbox{$G
  = (V,\gtran{},(V_{\Box},V_\Diamond,V_{\bigcirc}),\Prob)$} and a set of
target vertices $T$. Also, for every $n \in \Nset_0$ and every pair of
strategies $(\sigma,\pi) \in \Sigma \times \Pi$, let
$\calP_v^{\sigma,\pi}(\R_n(T))$ be the probability of all runs $w \in
\run(G(\sigma,\pi),v)$ such that $w$ visits some $u \in T$ in at most $n$
transitions (clearly, $\calP_v^{\sigma,\pi}(\R(T)) = \lim_{n \rightarrow
  \infty}\calP_v^{\sigma,\pi}(\R_n(T))$).

To keep this paper self-contained, we start by giving an elementary 
proof of Martin's weak determinacy result (see Equation~\ref{eq-determinacy}) 
for the special case of games with reachability objectives (observe that
the game $G$ fixed above is not required to be finite or finitely-branching). 

\begin{theorem}
\label{thm-value}
Every $v \in V$ has a value. Moreover, if $G$ is finitely-branching,
then there is a MD strategy $\pi \in \Pi$ which is optimal minimizing in
every vertex.
\end{theorem}
\begin{proof}
Let $(V\rightarrow [0,1], {\sqsubseteq})$ be the complete lattice of all
functions $f : V\rightarrow [0,1]$ with component-wise ordering. 
We show that the tuple of all values is the least fixed-point of the
following (Bellman) functional 
$\V:(V\rightarrow [0,1])\rightarrow (V\rightarrow [0,1])$ defined by
\begin{equation*}
\V(f)(v) \quad = \quad \begin{cases}
         1 & \text{ if } v\in T \\
         \sup\{f(u) \mid v \gtran{} u\} & \text{ if } v\in V_\Box\setminus T \\
         \inf\{f(u) \mid v \gtran{} u\} & \text{ if } v\in V_\Diamond\setminus T \\
         \sum_{v \gtran{x}u} x \cdot f(u) & \text{ if } v\in V_\bigcirc\setminus T 
         \end{cases}
       \end{equation*}
Since $\V$ is monotone, by Knaster-Tarski theorem~\cite{KT:fixed-complete}
there is the least fixed-point $\mu\V$ of $\V$.
Let  $\calA : V\rightarrow [0,1]$ be a function defined by
$\calA(v) = \sup_{\sigma \in \Sigma} \inf_{\pi \in \Pi} \calP_v^{\sigma,\pi}(\R(T))$. 
We prove the following:
\begin{enumerate}
\item[(i)] $\calA$ is a fixed point of $\V$.

\item[(ii)] For every $\varepsilon>0$ there is 
$\pi\in \Pi$ such that for every $v\in V$ we have that
\begin{equation}
\sup_{\sigma \in \Sigma} 
\calP_v^{\sigma,\pi}(\R(T)) \leq 
\mu\mathcal{V}(v)+\varepsilon
\end{equation}
\end{enumerate}
Observe that~(i) implies $\mu\V(v) \leq  
  \sup_{\sigma \in \Sigma} \inf_{\pi \in \Pi} \calP_v^{\sigma,\pi}(\R(T))$.
Obviously, 
\begin{equation*}
  \sup_{\sigma \in \Sigma} \inf_{\pi \in \Pi} \calP_v^{\sigma,\pi}(\R(T)) 
  \quad \leq \quad
  \inf_{\pi \in \Pi} \sup_{\sigma \in \Sigma} \calP_v^{\sigma,\pi}(\R(T)) 
\end{equation*}
and due to~(ii) we further have that 
$\inf_{\pi \in \Pi} \sup_{\sigma \in \Sigma} \calP_v^{\sigma,\pi}(\R(T))\leq \mu\V(v)$.
Hence, (i) and (ii) together imply that $\mu\V(v)$ is the value of~$v$
for every $v\in V$. It remains to prove (i) and~(ii).
\smallskip

Ad (i). Let $v\in V$. If $v\in T$, then clearly
$\mathcal{A}(v)=1=\mathcal{V}(\mathcal{A})(v)$.  If $v \not\in T$, we
can further distinguish three cases.
\begin{itemize}
\item[(a)] $v\in V_{\Box}$. Then
  \[
  \begin{array}{lclcl}
    \mathcal{V}(\mathcal{A})(v) 
    & = & \sup\{\mathcal{A}(u) \mid v \gtran{} u\}\\
    & = &
    \sup\{\sup_{\sigma \in \Sigma} \inf_{\pi \in \Pi} \calP_u^{\sigma,\pi}(\R(T)) \mid v \gtran{} u\} \\
    & = &
    \sup_{\sigma \in \Sigma} \inf_{\pi \in \Pi} \calP_v^{\sigma,\pi}(\R(T)) \\
    & = & \mathcal{A}(v) 
  \end{array}
  \]
\item[(b)] $v\in V_{\Diamond}$. Let us denote by $\dist(v)$ the set of
  all positive probability distributions on the set of outgoing
  edges of~$v$. Then
  \[
  \begin{array}{lll}
    \mathcal{V}(\mathcal{A})(v) 
    & = &  \inf\{\mathcal{A}(u) \mid v \gtran{} u\} \\
    & = & 
    \inf\{\sup_{\sigma \in \Sigma} \inf_{\pi \in \Pi} \calP_u^{\sigma,\pi}(\R(T)) \mid v \gtran{} u\} \\
    & = & 
    \inf_{\eta\in \dist(v)}\sum_{v \gtran{} u} \eta(v\gtran{} u)\cdot \sup_{\sigma \in \Sigma} \inf_{\pi \in 
      \Pi} \calP_u^{\sigma,\pi}(\R(T)) \\
    & =^{(*)} & 
    \sup_{\sigma \in \Sigma} \inf_{\eta\in \dist(v)}\sum_{v \gtran{} u} \eta(v\gtran{} u)\cdot 
    \inf_{\pi \in \Pi} 
    \calP_u^{\sigma,\pi}(\R(T)) \\
    & = & 
    \sup_{\sigma \in \Sigma} \inf_{\pi \in \Pi} 
    \calP_v^{\sigma,\pi}(\R(T)) \\
    & = & 
    \mathcal{A}(v)
  \end{array}
  \]
  In the equality~$(*)$,
  the `$\geq$' direction is easy, and the `$\leq$' direction can be
  justified as follows: For every $\delta>0$, there is a strategy
  $\bar{\sigma}\in \Sigma$ such that for every $u\in V$ we have that
  \[
  \sup_{\sigma \in \Sigma} \inf_{\pi \in \Pi}
  \calP_u^{\sigma,\pi}(\R(T)) \quad \leq \quad \inf_{\pi \in \Pi}
  \calP_u^{\bar{\sigma},\pi}(\R(T))+\delta
  \]
  This means that, for every $\eta\in \dist(v)$
  \[
  \sum_{v \gtran{} u} \eta(v\gtran{} u)\cdot \sup_{\sigma \in \Sigma}
  \inf_{\pi \in \Pi} \calP_u^{\sigma,\pi}(\R(T))\leq \sum_{v \gtran{}
    u} \eta(v\gtran{} u)\cdot \inf_{\pi \in \Pi}
  \calP_u^{\bar{\sigma},\pi}(\R(T))+\delta
  \]
  and thus  {\small
    \[
    \inf_{\eta\in \dist(v)}\sum_{v \gtran{} u} \eta(v\gtran{} u)\cdot
    \sup_{\sigma \in \Sigma} \inf_{\pi \in \Pi}
    \calP_u^{\sigma,\pi}(\R(T))\leq \inf_{\eta\in \dist(v)}\sum_{v
      \gtran{} u} \eta(v\gtran{} u)\cdot \inf_{\pi \in \Pi}
    \calP_u^{\bar{\sigma},\pi}(\R(T))+\delta
    \]} which implies $(*)$ %
  because $\delta$ was chosen arbitrarily.
\item[(c)] $v\in V_{\bigcirc}$. Then
  \[
  \begin{array}{lll}
    \mathcal{V}(\mathcal{A})(w) 
    & = &  \sum_{v \gtran{x}u} x\cdot \mathcal{A}(u) \\
    & = & 
    \sum_{v \gtran{x}u} x\cdot \sup_{\sigma \in \Sigma} \inf_{\pi \in \Pi} \calP_u^{\sigma,\pi}(\R(T)) \\
    & =^{(**)} & 
    \sup_{\sigma \in \Sigma}
    \inf_{\pi \in \Pi} \sum_{v \gtran{x}u} x\cdot \calP_u^{\sigma,\pi}(\R(T)) \\
    & = &
    \sup_{\sigma \in \Sigma}
    \inf_{\pi \in \Pi} \calP_v^{\sigma,\pi}(\R(T)) \\
    & = & \mathcal{A}(v)
  \end{array}
  \]
  Note that the equality~$(**)$ %
  can be justified similarly
  as~$(*)$ %
  above.
\end{itemize}

Ad (ii).  Let us fix some $\varepsilon>0$. For every $j\in\Nset_0$, we
define a strategy $\pi_j$ as follows: For a given $wv\in V^*
V_{\Diamond}$, we choose (some) edge $u\in V$ such that $\mu\V(u)\leq
\mu\V(v)+\frac{\eps}{2^{|w|+j+1}}$ and put $\pi_j(wv)(v\gtran{}
u)=1$. Note that such an edge must exist, and if $G$ is
finitely-branching, then there is even an edge $v \gtran{} u$ such
that $\mu\V(u) = \mu\V(v)$ (i.e., when $G$ is finitely-branching, we
can also consider the case when $\varepsilon = 0$).  In the sequel we
also write $\pi$ instead of $\pi_0$.  We prove that for all $\sigma\in
\Sigma$, $v\in V$, and $i\geq 0$ we have that
\[
\calP_v^{\sigma,\pi_j}(\R_i(T)) \quad \leq \quad
\mu\mathcal{V}(v)+\sum_{k=j+1}^{j+i} \frac{\varepsilon}{2^k}
\]
In particular, for $j=0$ we get
\[
  \calP_v^{\sigma,\pi}(\R_i(T)) \quad \leq \quad 
  \mu\mathcal{V}(v)+\sum_{k=1}^{i} \frac{\varepsilon}{2^k}
\]
and hence 
\[
\sup_{\sigma\in \Sigma}\calP_v^{\sigma,\pi}(\R(T)) \quad = \quad
\sup_{\sigma\in \Sigma}\lim_{i\rightarrow \infty}
\calP_v^{\sigma,\pi}(\R_i(T)) \quad \leq \quad  \mu\mathcal{V}(v)+\eps
\]
If $v\in T$, then 
$\calP_v^{\sigma,\pi_j}(\R_i(T))=1=\mu\mathcal{V}(v)$ for all $j\in\Nset_0$.
If $v\not \in T$, we proceed by induction on~$i$. If $i = 0$, then
$\calP_v^{\sigma,\pi_j}(\R_0(T))=0\leq \mu\V(v)$ for all $j\in\Nset_0$.
Now assume that $i\geq 1$. For every $\sigma\in\Sigma$, we use $\sigma_v$ 
to denote the strategy such that $\sigma_v(wu)=\sigma(vwu)$ for all
$wu\in\V^*\V_\Box$. We distinguish three cases. 
\begin{itemize}
\item[(a)]
  $v\in V_{\Box}$. Then
  \[
  \begin{array}{lcl}
    \calP_v^{\sigma,\pi_j}(\R_i(T))
    & = & 
    \sum_{v\gtran{} u}\sigma(v)(v\gtran{} u)\cdot
    \calP_u^{\sigma_v,\pi_{j+1}}(\R_{i-1}(T)) \\
    & \leq &
    \sum_{v\gtran{} u}\sigma(v)(v\gtran{} u)\cdot
    \left(\mu\mathcal{V}(u)+\sum_{k=j+2}^{j+i} \frac{\eps}{2^{k}}\right) \\
    & = &
    \left(\sum_{v\gtran{} u}\sigma(v)(v\gtran{} u)\cdot
      \mu\mathcal{V}(u)\right)+\sum_{k=j+2}^{j+i} \frac{\eps}{2^{k}} \\
    & \leq &
    \mu\V(v)+\sum_{k=j+2}^{j+i} \frac{\eps}{2^{k}}
  \end{array}
  \]
\item[(b)] $v\in V_{\Diamond}$. Then
  \[
  \begin{array}{lcl}
    \calP_v^{\sigma,\pi_j}(\R_i(T))
    & = & 
    \sum_{v\gtran{} u}\pi(v)(v\gtran{} u)\cdot
    \calP_u^{\sigma_v,\pi_{j+1}}(\R_{i-1}(T)) \\
    & \leq &
    \sum_{v\gtran{} u}\pi(v)(v\gtran{} u)\cdot
    \left(\mu\mathcal{V}(u)+\sum_{k=j+2}^{j+i} \frac{\eps}{2^{k}}\right) \\
    & = &
    \left(\sum_{v\gtran{} u}\pi(v)(v\gtran{} u)\cdot
      \mu\mathcal{V}(u)\right)+\sum_{k=j+2}^{j+i} \frac{\eps}{2^{k}} \\
    & \leq &
    \mu\V(v)+ \frac{\eps}{2^{j+1}} + \sum_{k=j+2}^{j+i} \frac{\eps}{2^{k}} \\
    & \leq &
    \mu\V(v)+ \sum_{k=j+1}^{j+i} \frac{\eps}{2^{k}}
  \end{array}
  \]           %
\item[(c)] $v\in V_{\bigcirc}$. Then
  \[
  \begin{array}{lcl}
    \calP_v^{\sigma,\pi_j}(\R_i(T))
    & = & 
    \sum_{v\gtran{x} u}x\cdot
    \calP_u^{\sigma_v,\pi_j}(\R_{i-1}(T)) \\
    & \leq &
    \sum_{v\gtran{x} u}x\cdot
    \left(\mu\mathcal{V}(u)+\sum_{k=j+2}^{j+i} \frac{\eps}{2^{k}}\right) \\
    & = &
    \left(\sum_{v\gtran{x} u}x\cdot
      \mu\mathcal{V}(u)\right)+\sum_{k=j+2}^{j+i} \frac{\eps}{2^{k}} \\
    & = &
    \mu\V(v)+\sum_{k=j+2}^{j+i} \frac{\eps}{2^{k}}
  \end{array}
  \]          
\end{itemize}

If $G$ is finitely branching,
then an optimal minimizing strategy $\pi$ is obtained by considering 
$\varepsilon=0$ in the above proof of~(ii).
\end{proof}

\begin{lemma}
\label{lem-max-bound} 
If $G$ is finitely-branching, then for every $v \in V$ we have that
\[ 
  \forall \varepsilon{>}0 \ \ \exists \sigma \in \Sigma \ \
  \exists n \in \Nset \ \  
  \forall \pi \in \Pi \ : \ \calP_v^{\sigma,\pi}(\R_n(T)) 
  > \val(v) - \varepsilon
\]
\end{lemma}
\begin{proof}
For all $v \in V$ and $i \in \Nset_0$, we use $\V_i(v)$ to denote the value 
of $v$ in $G$ with ``reachability in at most $i$-steps'' objective.
More precisely, we put $\V_{i}(v)=1$ for all $v \in T$ and $i\in\Nset_0$.
If $v \not\in T$, we define $\V_i(v)$ inductively as follows: $\V_0(v)=0$, 
and $\V_{i+1}(v)$ is equal either to $\max\{\V_i(v) \mid v \gtran{} u\}$,
$\min\{\V_i(u) \mid v \gtran{} u\}$, or $\sum_{v \gtran{x}u} x \cdot \V_i(u)$,
depending on whether $v \in V_\Box$, $v \in V_\Diamond$, or 
$v \in V_\bigcirc$, respectively. 

A straightforward induction on $i$ reveals that
\[
   \V_i(v) \quad = \quad 
   \max_{\sigma \in \Sigma} \min_{\pi \in \Pi} \calP_v^{\sigma,\pi}(\R_i(T)) 
\]
Also observe that, for every $i \in \Nset_0$, there is a fixed HD strategy
$\sigma_i \in \Sigma$ such that
for every $\pi \in \Pi$ and every $v\in V$ we have that
$\V_i(v) \leq \calP_v^{\sigma_i,\pi}(\R_i(T))$.
Further, put $\V_{\infty}(v) = \lim_{i \rightarrow \infty} \V_i(v)$
(note that the limit exists because the sequence
$\V_0(v),\V_1(v),\ldots$ is non-decreasing and bounded). We show that
$\V_{\infty}$ is a fixed point of the functional $\V$ defined in the
proof of Theorem~\ref{thm-value}. Hence, $\mu\V(v)\leq
\V_{\infty}(v)$ for every $v\in V$, which implies that for every
$\varepsilon>0$ there is $n\in \Nset$ such that for every $\pi\in \Pi$
we have that
\[
\calP_v^{\sigma_n,\pi}(\R_n(T))\geq V_n(v)>\mu\V(v)-\varepsilon=
\val(v) - \varepsilon
\]
So, it remains to prove that $\V(\V_{\infty})=\V_{\infty}$. We distinguish
three cases.
\begin{itemize}
\item[(a)] $v\in V_{\Box}$. Then 
\[
\V(\V_{\infty})(v)  =  \max_{v\gtran{} u} \lim_{i\rightarrow\infty} \V_i(u) 
                    =  \lim_{i\rightarrow \infty} \max_{v\gtran{} u} \V_i(u) 
                    =  \lim_{i\rightarrow \infty} \V_{i+1}(v)
                    =  \V_{\infty}(v)
\]
In the second equality, the `$\leq$' direction is easy, and the `$\geq$'
direction can be justified as follows: For every $u\in V$, the 
sequence $\V_1(u),\V_2(u),\ldots$ is non-decreasing.
Hence, for all $i\in \Nset$ and $u\in V$ we have that
$\lim_{j\rightarrow \infty} \V_j(u)\geq \V_i(u)$ and thus
$\max_{v\gtran{} u} \lim_{j\rightarrow \infty} \V_j(u)\geq \max_{v\gtran{} u}\V_i(u)$ 
which implies the `$\geq$' direction.
\item[(b)] $v\in V_{\Diamond}$. Then 
\[
\V(\V_{\infty})(v)  =  \min_{v\gtran{} u} \lim_{i\rightarrow \infty} \V_i(u) 
                    =  \lim_{i\rightarrow \infty} \min_{v\gtran{} u} \V_i(u) 
                    =  \lim_{i\rightarrow \infty} \V_{i+1}(v)
                    =  \V_{\infty}(v)
\]
In the second equality, the `$\geq$' direction is easy, and the
`$\leq$' direction can be justified as follows: For every $\delta>0$
there is $i\in \Nset$ such that for every $v\gtran{} u$ we have 
that $\lim_{j\rightarrow \infty} V_{j}(u)-\delta\leq
V_i(u)$ (remember that $G$ is finitely-branching).
It follows that
$\min_{v\gtran{} u}\lim_{j\rightarrow{} \infty} V_{j}(u)-\delta\leq \min_{v\gtran{} u} V_i(u)$ and thus
$\min_{v\gtran{} u}\lim_{j\rightarrow{} \infty} V_{j}(u)-\delta\leq 
\lim_{i\rightarrow{} \infty}\min_{v\gtran{} u} V_i(u)$ which implies the 
`$\leq$' direction because $\delta$ was chosen
arbitrarily.

\item[(c)] $v\in V_{\bigcirc}$. Then 
\[
\V(\V_{\infty})(v)  =  \sum_{v\gtran{x} u} x\cdot \lim_{i\rightarrow{} \infty} \V_i(u) 
                    =  \lim_{i\rightarrow \infty} \sum_{v\gtran{x} u} x\cdot \V_i(u) 
                    =  \lim_{i\rightarrow \infty} \V_{i+1}(v)
                    =  \V_{\infty}(v)
\]
by linearity of the limit.
\hfill\qed
\end{itemize}
\renewcommand{\qed}{}
\end{proof}

\noindent
Now we can state and prove the promised determinacy theorem. 

\begin{theorem}[Determinacy]
\label{thm-determinacy}
Assume that $G$ is finitely branching. Let 
$(T,{\succ}\varrho)$ be a (quantitative) 
reachability objective. Then $V = 
\wsetB{T}{\succ}{\varrho} \uplus \wsetD{T}{\notsucc}{\varrho}$.
\end{theorem}

\begin{proof}
  First, note that we may safely assume that for each $t \in T$
  there is only one out-going edge $t \gtran{} t$ (this assumption
  simplifies some of the claims presented below).
  Let $v \in V$. If $\varrho > \val(v)$, then $v \in
  \wsetB{T}{\succ}{\varrho}$ because player~$\Box$ has an
  \mbox{$\varepsilon$-optimal} strategy for an arbitrarily small
  $\varepsilon > 0$ (see Theorem~\ref{thm-value}). Similarly, if
  $\varrho < \val(v)$, then $v \in \wsetD{T}{\notsucc}{\varrho}$.  Now
  assume that $\varrho = \val(v)$. Obviously, it suffices to show that
  if player~$\Diamond$ does \emph{not} have a
  $(T,{\notsucc}\varrho)$-winning strategy in~$v$, then player~$\Box$
  has a $(T,{\succ}\varrho)$-winning strategy in~$v$. This means to
  show that
  \begin{equation}\label{eq:ap-impl-det-left}
    \forall \pi\in\Pi \ \ \exists \sigma\in \Sigma \ :\ 
    \calP^{\sigma,\pi}_{v}(\R(T))\succ\varrho
  \end{equation}
  implies
  \begin{equation*}%
    \exists \sigma\in\Sigma \ \forall \pi\in\Pi \ : \ 
    \calP^{\sigma,\pi}_{v}(\R(T))\succ\varrho
  \end{equation*}
  If $\succ$ is $>$ or $\val(v) = 0$, then the above implication
  follows easily. Observe that
  \begin{itemize}
  \item if $\succ$ is $>$, then (\ref{eq:ap-impl-det-left}) does not
    hold, because player~$\Diamond$ has an optimal minimizing strategy
    by Theorem~\ref{thm-value};
  \item for the constraint ${\geq}0$, the statement is trivial.
  \end{itemize}
  Hence, it suffices to consider the case when $\succ$ is $\geq$ and
  $\varrho=\val(v)>0$. Assume that~(\ref{eq:ap-impl-det-left}) holds.
  We say that a vertex $u \in V$ is \emph{good} if 
  \begin{equation}\label{eq:good}
    \forall \pi\in\Pi \ \ \exists \sigma\in \Sigma \ :\ 
    \calP^{\sigma,\pi}_{u}(\R(T)) \geq \val(u)
  \end{equation}
  Note that the vertex $v$ fixed above is good by~(\ref{eq:ap-impl-det-left}).
  Further, we say that an edge $u \gtran{} u'$ of $G$ is \emph{optimal}
  if either $u \in V_{\bigcirc}$, or $u \in V_{\Box} \cup V_{\Diamond}$ and 
  $\val(u) = \val(u')$.
  Observe that for every $u \in V_{\Box} \cup V_{\Diamond}$ there is at least
  one optimal edge $u \gtran{} u'$, because $G$ is finitely branching
  (recall that the tuple of all values
  is the least fixed-point of the functional $\V$ defined in
  the proof of Theorem~\ref{thm-value}). Further,
  note that if $u \in V_{\Box}$ is a good vertex, then there is at least
  one optimal edge $u \gtran{} u'$ where $u'$ is good (otherwise
  we immediately obtain a contradiction with~(\ref{eq:good}); also
  observe that if $u \in T$, then $u \gtran{} u$ by the technical assumption
  above).
  Similarly, if $u \in V_{\Diamond}$ is good then for every optimal edge
  $u \gtran{} u'$ we have that $u'$ is good, and if   
  $u \in V_{\bigcirc}$ is good and $u \gtran{} u'$ then $u'$ is good.
  Hence, we can define a game $\bar{G}$, where the set of vertices 
  $\bar{V}$ consists of all good vertices of~$G$, and for all 
  $u,u' \in \bar{V}$ we have that $(u,u')$ is
  an edge of $\bar{G}$ iff $u \gtran{} u'$ is an optimal edge of~$G$.
  The edge probabilities in $\bar{G}$ are the same as in~$G$.
  The rest of the proof proceeds by proving the following three claims:
  \begin{enumerate}
  \item[(a)] For every $u \in \bar{V}$  we have that
    $\val(u,\bar{G}) = \val(u,G)$.
  \item[(b)] There is $\bar{\sigma} \in \Sigma_{\bar{G}}$
    such that for every
    $\bar{\pi} \in \Pi_{\bar{G}}$ we have that
    $\calP^{\bar{\sigma},\bar{\pi}}_{v}(\R(T,\bar{G})) \geq 
    \val(v,\bar{G}) = \varrho$.
  \item[(c)] The strategy $\bar{\sigma}$ can be modified
    into a strategy $\sigma \in \Sigma_{G}$ such that
    for every
    $\pi \in \Pi_G$ we have that
    $\calP^{\sigma,\pi}_{v}(\R(T,G)) \geq \varrho$.
  \end{enumerate}
  We start by proving Claim~(a). Let $u \in \bar{V}$. 
  Due to Theorem~\ref{thm-value},
  there is a MD strategy $\pi \in \Pi_G$ which
  is optimal minimizing in every vertex of $G$ (particularly 
  in~$u$) and selects only the optimal edges. Hence, the strategy
  $\pi$ can also be used in the restricted game $\bar{G}$ and thus
  we obtain $\val(u,\bar{G}) \leq \val(u,G)$. Now suppose that
  $\val(u,\bar{G}) < \val(u,G)$. By applying 
  Theorem~\ref{thm-value} to $\bar{G}$, there is an 
  optimal minimizing MD strategy $\bar{\pi} \in \Pi_{\bar{G}}$.
  Further, for every vertex $t$ of $G$ which is not good there
  is a strategy $\pi_t \in \Pi_G$ such that for every 
  $\sigma \in \Sigma_G$ we have that 
  $\calP^{\sigma,\pi_t}_{t}(\R(T,G)) < \val(t,G)$ (this follows immediately
  from~(\ref{eq:good})).
  Now consider a strategy $\pi' \in \Pi_G$ which for every play of $G$ 
  initiated in $u$ behaves in the following way:
  \begin{itemize}
  \item As long as player~$\Box$ uses only the edges of $G$ that are 
    preserved in $\bar{G}$, the strategy $\pi'$ behaves exactly like the 
    strategy $\bar{\pi}$.
  \item When player~$\Box$ uses an edge $r \gtran{} r'$ which 
    is not an edge in $\bar{G}$ for the first time, then the strategy
    $\pi'$ starts to behave either like the optimal minimizing strategy
    $\pi$ or the strategy $\pi_{r'}$, depending on whether $r'$ is good 
    or not (observe that if $r'$ is good, then $\val(r',G) < \val(r,G)$).
  \end{itemize}
  Now it is easy to check that for every $\sigma \in \Sigma_{G}$ we
  have that $\calP^{\sigma,\pi'}_{u}(\R(T,G)) < \val(u,G)$, which 
  contradicts the assumption that $u$ is good.  

  Now we prove Claim~(b).  Due to Lemma~\ref{lem-max-bound}, for every
  $u \in \bar{V}$ we can fix a strategy 
  $\bar{\sigma}_u \in \Sigma_{\bar{G}}$ and 
  $n_u \in \Nset$ such that for every $\bar{\pi} \in \Pi_{\bar{G}}$ we have that
  $\calP^{\bar{\sigma}_u,\bar{\pi}}_{u}(\R_{n_u}(T,\bar{G}))> \val(u,\bar{G})/2$.
  For every $k \in \Nset_0$, let $B(k)$ be the set of all vertices $u$
  reachable from $v$ in $\bar{G}$ via a path of length exactly~$k$
  which does not visit~$T$. Observe that $B(k)$ is finite because
  $\bar{G}$ is finitely-branching. Further, for every $i \in \Nset_0$
  we define a bound $m_i \in \Nset$ inductively as follows: $m_0 = 1$,
  and $m_{i+1} = m_i + \max\{n_{u} \mid u \in B(m_i)\}$.  Now we
  define a strategy $\bar{\sigma} \in \Sigma_{\bar{G}}$ which turns out to be
  $(T,{\geq}\varrho)$-winning in the vertex $v$ of $\bar{G}$. 
  For every $w \in \bar{V}^*\bar{V}_\Box$  such that $m_i \leq |w| < m_{i+1}$ we 
  put $\bar{\sigma}(w) = \bar{\sigma}_u(uw_2)$, where 
  $w=w_1 u w_2$, $|w_1|=m_i-1$ and $u \in \bar{V}$.
  Now it is easy to check that for every $i \in \Nset$ and every
  strategy $\bar{\pi} \in \Pi_{\bar{G}}$ we have that
  $\calP^{\bar{\sigma},\bar{\pi}}_{v}(\R_{m_i}(T,\bar{G})) >
  (1-\frac{1}{2^i})\varrho$.  This means that the strategy $\bar{\sigma}$ is
  $(T,{\geq}\varrho)$-winning in~$v$.

  It remains to prove Claim~(c). Consider a strategy $\sigma \in \Sigma_G$
  which for every play of $G$ initiated in $v$ behaves as follows:
  \begin{itemize}
  \item As long as player~$\Diamond$ uses only the optimal edges, the
    strategy $\sigma$ behaves exactly like the strategy $\bar{\sigma}$.
  \item When player~$\Diamond$ uses a non-optimal edge $r \gtran{} r'$
    for the first time, the strategy $\sigma$ starts to behave 
    like an $\varepsilon$-optimal maximizing strategy in $r'$,
    where $\varepsilon = (\val(r',G) - \val(r,G))/2$. Note that
    since $r \gtran{} r'$ is not optimal, we have that
    $\val(r',G) > \val(r,G)$.
  \end{itemize}
  It is easy to check that $\sigma$ is $(T,{\geq}\varrho)$-winning in~$v$.
\end{proof}

\section{Stochastic BPA Games}
\label{sec-BPA-games}

Stochastic BPA games correspond to stochastic games induced by 
stateless pushdown automata or 1-exit recursive state machines
(see Section~\ref{sec-intro}). A formal definition follows.

\begin{definition}
  A \emph{stochastic BPA} game is a tuple 
  $\Delta = (\Gamma,\btran{},(\Gamma_\Box,\Gamma_\Diamond,\Gamma_\bigcirc),
  \Prob)$ where $\Gamma$ is a finite \emph{stack alphabet},
  ${\btran{}} \subseteq \Gamma \times \Gamma^{\leq 2}$ is a finite set
  of \emph{rules} (where $\Gamma^{\leq 2} = \{w \in \Gamma^* : |w| \leq 2\}$)
  such that for each $X \in \Gamma$ there is some rule $X \btran{} \alpha$,
  $(\Gamma_\Box,\Gamma_\Diamond,\Gamma_\bigcirc)$ is a partition of $\Gamma$,
  and $\Prob$ is a \emph{probability assignment} which to each 
  $X \in \Gamma_\bigcirc$ assigns a rational positive probability distribution
  on the set of all rules of the form $X \btran{} \alpha$.
\end{definition}

A \emph{configuration} of $\Delta$ is a word $\alpha \in \Gamma^*$, which
can intuitively be interpreted as the current stack content where
the leftmost symbol of $\alpha$ is on top of the stack.
Each stochastic BPA game $\Delta = 
(\Gamma,\btran{},(\Gamma_{\Box},\Gamma_\Diamond,\Gamma_{\bigcirc}),\Prob)$
determines a unique stochastic game 
\mbox{$G_\Delta = 
(\Gamma^*,\gtran{},(\Gamma_{\Box}\Gamma^*,\Gamma_{\Diamond}\Gamma^*,
\Gamma_{\bigcirc}\Gamma^* \cup \{\varepsilon\}),\Prob_\Delta)$},
where the edges of $\gtran{}$ are determined as follows:
$\varepsilon \gtran{} \varepsilon$, and
$X\beta \gtran{} \alpha\beta$ iff $X \btran{} \alpha$. The probability
assignment $\Prob_\Delta$ is the natural extension of $\Prob$, i.e.,
$\varepsilon \gtran{1} \varepsilon$ and
for all $X \in \Gamma_{\bigcirc}$
we have that 
$X\beta \gtran{x} \alpha\beta$ iff  $X \btran{x} \alpha$. 
The size of $\Delta$, denoted by $|\Delta|$, is the length of the
corresponding binary encoding.

In this section we consider stochastic BPA games with qualitative 
reachability objectives $(T,{\succ}\varrho)$ where $T \subseteq \Gamma^*$ 
is a \emph{regular} set of configurations. For technical convenience, 
we define the size of $T$ as the size of the minimal deterministic 
finite-state automaton $\A_T = (Q,q_0,\delta,F)$ which recognizes
the \emph{reverse} of $T$ (if we view configurations as stacks, this 
corresponds to the bottom-up direction). Note that the automaton
$\A_T$ can be simulated on-the-fly in $\Delta$ by employing standard 
techniques (see, e.g., \cite{EKS:PDA-regular-valuations-IC}). That is,
the stack alphabet is extended to $\Gamma \times Q$ and the rules
are adjusted accordingly (for example, if $X \btran{} YZ$, then for every
$q \in Q$ the extended BPA game has a rule 
$(X,q) \btran{} (Y,r)(Z,q)$ where $\delta(q,Z) = r$).
Note that the on-the-fly simulation of $\A_T$ in $\Delta$ does not affect 
the way how the game is played, and the size of the extended game is polynomial
in $|\Delta|$ and $|\A_T|$. The main advantage of this simulation is
that the information whether a current configuration belongs
to $T$ or not can now be deduced just by looking at the symbol on top
of the stack. This leads to an important technical simplification
in the definition of $T$.

\begin{definition}
  We say that $T \subseteq \Gamma^*$ is \emph{simple} if
  $\varepsilon \notin T$ and there is
  $\Gamma_T \subseteq \Gamma$ such that for every $X\alpha \in \Gamma^+$
  we have that $X\alpha \in T$ iff $X \in \Gamma_T$.
\end{definition}

Note that the requirement $\varepsilon \notin T$ in the previous 
definition is not truly restrictive, because each BPA can be equipped
with a fresh bottom-of-the-stack symbol which cannot be removed. Hence,
we can safely restrict ourselves just to simple sets of target 
configurations. All of the obtained results (including the complexity
bounds) are valid also for regular sets of target configurations.

Since stochastic BPA games have infinitely many vertices, even 
memoryless strategies are not necessarily finitely representable.
It turns out that the winning strategies for both players in
stochastic BPA games with qualitative reachability objectives
are (effectively) \emph{regular} in the following sense:

\begin{definition}
\label{def-regular-strategy}
  Let $\Delta = (\Gamma,\btran{},(\Gamma_\Box,\Gamma_\Diamond,\Gamma_\bigcirc),
  \Prob)$ be a stochastic BPA game, and let ${\odot} \in \{\Box,\Diamond\}$.
  We say that a strategy $\tau$ for player~$\odot$
  is \emph{regular} if there is a deterministic finite-state automaton
  $\A$ over the alphabet $\Gamma$ such that,
  for every $X\alpha \in \Gamma_\odot\Gamma^*$, the value of 
  $\tau(X\alpha)$ depends just on the control state entered by
  $\A$ after reading the reverse of $X\alpha$ (i.e., the automaton
  $\A$ reads the stack bottom-up). Note that regular strategies are
  not necessarily deterministic.

  A special type of regular strategies are \emph{stackless MD (SMD)} 
  strategies, where $\tau(X\alpha)$ depends just on the symbol $X$
  on top of the stack. Note that SMD strategies are deterministic.

\end{definition}

We use $T_\varepsilon$ to denote the 
set $T \cup \{\varepsilon\}$, and we also slightly abuse the notation 
by writing $\varepsilon$ instead of $\{\varepsilon\}$ (particularly 
in expressions such as $\R(\varepsilon)$ or $\wsetD{\varepsilon}{<}{1}$).

In the next sections, we consider the two meaningful qualitative
probability constraints ${>}0$ and ${=}1$. We show that the winning
regions $\wsetB{T}{>}{0}$, $\wsetD{T}{=}{0}$, $\wsetB{T}{=}{1}$, and
$\wsetD{T}{<}{1}$ are effectively regular. Further, we show that 
the membership to $\wsetB{T}{>}{0}$ and $\wsetD{T}{=}{0}$ is in
$\PTIME$, and the membership to $\wsetB{T}{=}{1}$ and $\wsetD{T}{<}{1}$
is in  $\NP \cap \coNP$. Finally, we show that the associated winning
strategies are regular and effectively constructible (for both players).

\section{Computing the Regions 
    \protect{$\wsetB{T}{>}{0}$} and 
    \protect{$\wsetD{T}{=}{0}$}} 
\label{sec:zero}

For the rest of this section, we fix a stochastic BPA game
\mbox{$\Delta = (\Gamma,\btran{},(\Gamma_\Box,\Gamma_\Diamond,%
\Gamma_\bigcirc),\Prob)$}
and a simple set $T$ of target configurations. Since we are interested
only in reachability objectives, we can safely assume that for
every $R \in \Gamma_T$, the only rule where $R$ appears on the
left-hand side is $R \btran{} R$ (this assumption simplifies the
formulation of some claims). 

We start by observing that the sets  $\wsetB{T}{>}{0}$ and
$\wsetD{T}{=}{0}$ are regular,
and the associated finite-state automata have a fixed number of
control states.

\begin{proposition}
\label{prop-greater-zero-regular}
  Let $\A = \wsetB{T}{>}{0} \cap \Gamma$ and 
  $\B = \wsetB{T_\varepsilon}{>}{0} \cap \Gamma$. Then
  $\wsetB{T}{>}{0} = \B^*\A \Gamma^*$ and 
  $\wsetB{T_\varepsilon}{>}{0} = \B^*\A\Gamma^* \cup \B^*$.
  Consequently, $\wsetD{T}{=}{0} = \Gamma^* \smallsetminus \wsetB{T}{>}{0} =
  (\B \smallsetminus \A)^* \cup (\B \smallsetminus \A)^*(\Gamma 
  {\smallsetminus} \B)\Gamma^*$ and
  $\wsetD{T_\varepsilon}{=}{0} = \Gamma^* \smallsetminus 
  \wsetB{T_\varepsilon}{>}{0} = (\B \smallsetminus \A)^*(\Gamma 
  {\smallsetminus} \B)\Gamma^*$.
\end{proposition}
\begin{proof}
Note that $\A\subseteq\B$. We start by introducing some notation. 
For every strategy $\sigma \in \Sigma$ and every $\alpha \in \Gamma^*$, let
\begin{itemize}
\item $\sigma[{-}\alpha]$ be a strategy such that for every
  finite sequence of configurations
  $\gamma_1,\ldots,\gamma_n,\gamma$, where $n \geq 0$ and 
  $\gamma \in \Gamma_\Box\Gamma^*$, and every edge
  $\gamma \gtran{} \delta$ we have that 
  $\sigma[{-}\alpha](\gamma_1,\ldots,\gamma_n,\gamma)(\gamma \gtran{} \delta) = 
   \sigma(\gamma_1\alpha,\ldots,\gamma_n\alpha,\gamma\alpha)(\gamma\alpha \gtran{} \delta\alpha)$
\item $\sigma[{+}\alpha]$ be a strategy such that for every
  finite sequence of configurations
  $\gamma_1\alpha,\ldots,\gamma_n\alpha,\gamma\alpha$, where  $n \geq 0$ and
   $\gamma\alpha \in \Gamma_\Box\Gamma^*$, and every edge
  $\gamma\alpha \gtran{} \delta\alpha$ we have that 
  $\sigma[{+}\alpha](\gamma_1\alpha,\ldots,\gamma_n\alpha,\gamma\alpha)(\gamma\alpha \gtran{} \delta\alpha) = 
   \sigma(\gamma_1,\ldots,\gamma_n,\gamma)(\gamma \gtran{} \delta)$
\end{itemize}

By induction on the length of  $\alpha\in\Gamma^*$,
we prove that $\alpha\in\wsetB{T}{>}{0}$ iff $\alpha\in\B^*\A\Gamma^*$.
For $\alpha=\varepsilon$, both sides of the equivalence are false.
Now assume that the equivalence holds for all configurations of
length $k$ and consider an arbitrary $X\alpha\in\Gamma^+$ where
$|\alpha|=k$. If $X\alpha\in\wsetB{T}{>}{0}$, then there are two
possibilities:
\begin{itemize}
\item There is a strategy $\sigma\in\Sigma$ such that for all
  $\pi\in\Pi$, the probability of reaching $T$ without prior reaching
  $\alpha$ is positive in the play $G_\Delta(\sigma,\pi)$
  initiated in $X\alpha$.  Then
  $\sigma[{-}\alpha]$ is $(T,{>}0)$-winning in $X$, which 
  means that $X\in\wsetB{T}{>}{0}$, i.e., $X\in\A$.
\item There is a strategy $\sigma\in\Sigma$ such that for all
  $\pi\in\Pi$, the probability of reaching $T$ is positive
  in the play $G_\Delta(\sigma,\pi)$ initiated in $X\alpha$,
  but for some $\hat{\pi} \in \Pi$, the configuration $\alpha$ is
  always reached before reaching $T$. In this case, consider again the
  strategy $\sigma[{-}\alpha]$.
  Then $\sigma[{-}\alpha]$ is $(T_\varepsilon,{>}0)$-winning
  in $X$, which means
  $X\in\wsetB{T_\varepsilon}{>}{0}$, i.e., $X\in\B$.  Moreover, observe
  that the strategy $\sigma$ is $(T,{>}0)$-winning in $\alpha$.
  Thus, $\alpha\in\wsetB{T}{>}{0}$ and by induction hypothesis we obtain
  $\alpha\in\B^*\A\Gamma^*$.
\end{itemize}
In both cases, we obtained $X\alpha\in\B^*\A\Gamma^*$.
If $X\alpha\in\B^*\A\Gamma^*$, we can again distinguish two possibilities: 
\begin{itemize}
\item  $X\in\A$ and there is a $(T,{>}0)$-winning strategy
$\sigma\in\Sigma$ for the initial configuration $X$.
Then the strategy $\sigma[{+}\alpha]$ is $(T,{>}0)$-winning in $X\alpha$. 
Thus, $X\alpha\in\wsetB{T}{>}{0}$.

\item  $X\in\B$ and $\alpha\in\B^*\A\Gamma^*$. Then there exists a
$(T_\varepsilon,{>}0)$-winning strategy $\sigma_1\in\Sigma$ in $X$.
By induction hypothesis, there is a 
$(T,{>}0)$-winning strategy $\sigma_2\in\Sigma$ in $\alpha$.
We construct a strategy $\sigma'$ which behaves like $\sigma_1[{+}\alpha]$
until $\alpha$ is reached, and from that point on it behaves like
$\sigma_2$. Obviously, $\sigma'$ is $(T,{>}0)$-winning, which means
that $X\alpha\in\wsetB{T}{>}{0}$.
\end{itemize}
The proof of $\wsetB{T_\varepsilon}{>}{0} = \B^*\A\Gamma^* \cup \B^*$
is similar.
\end{proof}

Our next proposition says how to compute the sets $\A$ and $\B$. 

\begin{proposition}
\label{prop-greater-than-zero}
The pair $(\A,\B)$ is  the least fixed-point
of the function \mbox{$F : (2^\Gamma \times 2^\Gamma) \rightarrow
(2^\Gamma \times 2^\Gamma)$} defined as follows: 
$F(A,B) = (\hat{A},\hat{B})$, where
\begin{eqnarray*}
  \hat{A} & = & \Gamma_T \cup A  \cup 
                \{X \in \Gamma_\Box \cup \Gamma_\bigcirc \mid 
                  \mbox{ there is } X \btran{} \beta 
                  \mbox{ such that } \beta \in B^*A\Gamma^*\}\\
         & \cup & \{X \in \Gamma_\Diamond  \mid 
                  \mbox{ for all } X \btran{} \beta 
                  \mbox{ we have that } \beta \in B^*A\Gamma^*\}\\
  \hat{B} & = & \Gamma_T \cup B \cup 
                \{X \in \Gamma_\Box \cup \Gamma_\bigcirc \mid 
                  \mbox{ there is } X \btran{} \beta 
                  \mbox{ such that } \beta \in B^*A\Gamma^* \cup B^*\}\\
         & \cup & \{X \in \Gamma_\Diamond  \mid 
                  \mbox{ for all } X \btran{} \beta 
                  \mbox{ we have that } \beta \in B^*A\Gamma^* \cup B^*\}
\end{eqnarray*}
\end{proposition}
\begin{proof}
For every $i \in \Nset_0$, let $(A_i,B_i) = F^i(\emptyset,\emptyset)$.
The set $2^\Gamma \times 2^\Gamma$ with the component-wise inclusion
forms a finite lattice. The longest chain in this lattice has length
$2|\Gamma|+1$.
Since $F$ is clearly monotone, by Knaster-Tarski theorem 
$(\A_F,\B_F) = (\bigcup_{i=0}^{2|\Gamma|} A_i,\bigcup_{i=0}^{2|\Gamma|} B_i)$
is the least fixed-point of $F$.
We show that 
$(\A_F,\B_F) = (\A,\B)$.

We start with the ``$\subseteq$'' direction. We use the following notation:
\begin{itemize}
\item for every $X \in \A_F$, let $I_A(X)$ be the 
  least $i \in \Nset$ such that $X \in A_i$;
\item for every $X \in \B_F$, let $I_B(X)$ be the 
  least $i \in \Nset$ such that $X \in B_i$;
\item for every $\alpha Y \in \B_F^* \A_F$, let 
  $I(\alpha Y) = \max (\{I_A(Y)\} \cup \{I_B(Z) \mid Z 
  \mbox{ appears in } \alpha\})$;
\item for every $\beta \in \Gamma^*$, let
  $\price(\beta) = \min\{I(\gamma) \mid  \gamma \mbox{ is a prefix of } 
  \beta, \gamma \in \B_F^* \A_F\}$, where $\min(\emptyset) {=} \infty$.
\end{itemize}

First observe that $\Gamma_T$ is a subset of both $\A$ and $\B$.
For every $X \in (\A_F \cap \Gamma_\Box) \smallsetminus \Gamma_T$, we fix some
$X \btran{} \alpha$  (the ``$A$-rule'') such that  $\price(\alpha) < I_A(X)$.
It follows directly from the definition
of $F$ that there must by such a rule. Similarly, for every
$X \in (\B_F \cap \Gamma_\Box) \smallsetminus \Gamma_T$, we fix some
$X \btran{} \alpha$ (the ``$B$-rule'') such that either 
$\price(\alpha) < I_B(X)$, or
$\alpha \in \B_F^*$ and $I_B(Y) < I_B(X)$ for every $Y$ of $\alpha$.

Now consider a MD strategy $\sigma \in \Sigma$ which for a given
$X\alpha \in \B_F^*\A_F\Gamma^* \cap \Gamma_\Box \Gamma^*$ selects
\begin{itemize}
\item an arbitrary outgoing rule if $X\in \Gamma_T$;
\item the $A$-rule of $X$ if $X \in \A_F$ and $I_A(X) = \price(X\alpha)$;
\item the $B$-rule of $X$ otherwise.
\end{itemize}
We claim that $\sigma$ is $(T,{>}0)$-winning in 
every configuration
of $\B_F^*\A_F\Gamma^*$. In particular, this
means that $\A_F \subseteq \A$. To see this, realize that 
for every $\pi \in \Pi$, the play $G_\Delta(\sigma,\pi)$ contains
a path along which every transition either decreases the price, 
or maintains the price but decreases either the length or replaces the
first symbol with a sequence of symbols whose 
$I_B$-value is strictly smaller. Hence, this
path must inevitably visit $T$ after performing a finite number 
of transitions.

Similar arguments show that
$\sigma$ is $(T,{>}0)$-winning 
in every configuration
of $\B_F^*\A_F\Gamma^* \cup \B_F^*$. In particular, this
means that $\B_F \subseteq \B$.

Now we prove the ``$\supseteq$'' direction, i.e., $\A_F\supseteq\A$ and
$\B_F\supseteq\B$.  Let us define the \emph{$\A$-norm} of a given $X\in
\Gamma$, $N_A(X)$, to be the least $n$ such that for some $\sigma\in \Sigma$
and for all $\pi\in \Pi$ there is a path in $G_\Delta(\sigma,\pi)$ of length at
most $n$ from $X$ to $T$.
Similarly, define the \emph{$\B$-norm} of a given $X\in \Gamma$, $N_B(X)$,
to be the least $n$ such that for some $\sigma\in \Sigma$ and for all
$\pi\in \Pi$ there is a path in $G_\Delta(\sigma,\pi)$ of length at most $n$
from $X$ to $T_\varepsilon$ (if there are no such paths, then we put
$N_A(X)=\infty$ and $N_B(X)=\infty$, respectively).

It follows from K\"onig's lemma and the fact that the game is finitely
branching that $N_A(X)$ is finite for every $X\in\A$, and $N_B(X)$ is finite
for every $X\in\B$. Also note that for all $X\in\Gamma$ we have that
$N_A(X)\geq N_B(X)$.

We show, by induction on $n$, 
that every $X\in\A$ s.t.{} $N_A(X)=n$ belongs to $A_n$, and
that every $X\in\B$ s.t.{} $N_B(X)=n$ belongs to $B_n$.
The base case is easy since
$N_A(X)=1$ iff $N_B(X)=1$ iff $X\in\Gamma_T$,
and $(A_1,B_1)=(\Gamma_T,\Gamma_T)$.
The inductive step follows:
\begin{itemize}
\item $X\in\A$. If $X\in\Gamma_\Box$ (or $X\in\Gamma_\Diamond$), then some
  (or every) rule of the form $X\btran{} \beta Y \gamma$ satisfies
  $\beta\in \B^*$, $Y\in \A$, $N_A(Y)<n$, and $N_B(Z)<n$ for 
  all $Z$ which appear in $\beta$.  By induction hypothesis, 
  $\beta\in B_{n-1}^*$ and $Y\in A_{n-1}$. Hence, $X\in A_n$.
\item $X\in\B$. If $X\in\Gamma_\Box$ (or $X\in\Gamma_\Diamond$), then some
  (or every) rule of the form $X\btran{}\bar{\beta}$ satisfies
  one of the following conditions:
  \begin{itemize}
  \item $\bar{\beta}=\beta Y \gamma$ where $\beta\in \B^*$, 
    $Y\in \A$, $N_A(Y)<n$, and $N_B(Z)<n$ for all $Z$ which 
    appear in $\beta$.  By induction hypothesis,
    $\beta\in B_{n-1}^*$ and $Y\in A_{n-1}$. Hence, $X\in A_n\subseteq B_n$.
  \item $\bar{\beta}\in \B^*$ where $N_B(Z)<n$ for all $Z$ which appear
  in $\bar\beta$.
    By induction hypothesis, $\bar\beta\in B_{n-1}^*$, and hence $X\in B_n$.
  \end{itemize}
\end{itemize}
\end{proof}

Since the least fixed-point of the function~$F$ defined in
Proposition~\ref{prop-greater-than-zero} is computable in polynomial
time, the finite-state automata recognizing the sets 
$\wsetB{T}{>}{0}$ and $\wsetD{T}{=}{0}$ are computable in polynomial
time. Thus, we obtain the following theorem:

\begin{theorem}
\label{thm-BPA-case-zero}
  The membership to $\wsetB{T}{>}{0}$ and $\wsetD{T}{=}{0}$ is decidable 
  in polynomial time. Both sets are effectively regular, and the associated
  finite-state automata are constructible
  in polynomial time.
  Further, there is a regular strategy $\sigma \in \Sigma$ 
  and a SMD strategy $\pi \in \Pi$ constructible in polynomial time
  such that $\sigma$ and $\pi$ is 
  $(T,{>}0)$-winning and $(T,{=}0)$-winning in every
  configuration of $\wsetB{T}{>}{0}$ and $\wsetD{T}{=}{0}$, respectively. 
\end{theorem}
\begin{proof}
Due to Proposition~\ref{prop-greater-than-zero}, it only remains to show
that $\sigma$ is regular, $\pi$ is SMD, and both $\sigma$ and $\pi$
are effectively constructible in polynomial time. Observe
that the MD strategy $\sigma$
defined in the proof of Proposition~\ref{prop-greater-than-zero}
is $(T,{>}0)$-winning for player~$\Box$. Moreover, $\sigma$
is regular, because the $\price$ of a given configuration can 
be determined by an effectively constructible finite-state 
automaton which reads configurations from right to left. Since
the $\price$ of a given configuration is bounded by $2|\Gamma|$, 
the automaton needs only $\calO(|\Gamma|)$ control states and can
be easily computed in polynomial time.

A SMD $(T,{=}0)$-winning strategy $\pi$ for  player~$\Diamond$
is easy to construct. Consider a strategy $\pi$ such that for
every $X\alpha \in \Gamma_\Diamond\Gamma^*$ we have that
\begin{itemize}
\item if $X\in(\B\smallsetminus\A)$, then $\pi(X\alpha)$ selects 
  an edge $X\alpha \gtran{} \beta\alpha$ where $X\btran{} \beta$ and
  $\beta\in(\B\smallsetminus\A)^*\cup(\B\smallsetminus\A)^*(\Gamma\smallsetminus\B)\Gamma^*$;
\item if $X\in(\Gamma\smallsetminus\B)$, then $\pi(X\alpha)$ selects 
  an edge $X\alpha \gtran{} \beta\alpha$ where $X\btran{} \beta$ and
  $\beta\in(\B\smallsetminus\A)^*(\Gamma\smallsetminus\B)\Gamma^*$;
\item otherwise, $\pi$ is defined arbitrarily.
\end{itemize}
It is easy to check that $\pi$ is $(T,{=}0)$-winning in every configuration
of $\wsetD{T}{=}{0}= (\B\smallsetminus\A)^*\cup
 (\B\smallsetminus\A)^*(\Gamma\smallsetminus\B)\Gamma^*$.
\end{proof}

\begin{remark}
\label{rem-extend-case-zero}
Note that Theorem~\ref{thm-BPA-case-zero} holds also for the winning regions
$\wsetB{T_\varepsilon}{>}{0}$ and $\wsetD{T_\varepsilon}{=}{0}$. The argument
is particularly simple in the case of $\wsetD{T_\varepsilon}{=}{0}$, where
we only need to modify the strategy $\pi$ constructed in the proof of
Theorem~\ref{thm-BPA-case-zero} so that if
$X \in (\B \smallsetminus \A)$, then $\pi(X\alpha)$ selects an edge
$X\alpha \gtran{} \beta\alpha$ where $X \btran{} \beta$ and
$\beta \in (\B\smallsetminus\A)^*(\Gamma\smallsetminus\B)\Gamma^*$.
\end{remark}

\section{Computing the Regions 
   \protect{$\wsetB{T}{=}{1}$} and 
   \protect{$\wsetD{T}{<}{1}$}} 
\label{sec:one}

The results presented in this subsection constitute the very
core of this paper. The problems are more complicated
than in the case of $\wsetB{T}{>}{0}$ and $\wsetD{T}{=}{0}$, 
and several deep observations
are needed to tackle them. As in Section~\ref{sec:zero}, we 
fix a stochastic BPA game
\mbox{$\Delta = (\Gamma,\btran{},(\Gamma_\Box,\Gamma_\Diamond,%
\Gamma_\bigcirc),\Prob)$} and a simple set $T$ of target configurations
such that, for every $R \in \Gamma_T$, the only rule where $R$ 
appears on the left-hand side is $R \btran{} R$.

The regularity of the
sets $\wsetB{T}{=}{1}$ and $\wsetD{T}{<}{1}$ is revealed in the
next proposition.

\begin{proposition}
\label{prop-equal-one-regular}
  Let $\A = \wsetD{T_\varepsilon}{<}{1} \cap \Gamma$,
  $\B = \wsetB{T_\varepsilon}{=}{1} \cap \Gamma$,
  $\C = \wsetD{T}{<}{1} \cap \Gamma$, and
  $\scrD = \wsetB{T}{=}{1} \cap \Gamma$. Then
  $\wsetB{T}{=}{1} = \B^*\scrD \Gamma^*$ and 
  $\wsetD{T}{<}{1} = \C^*\A\Gamma^* \cup \C^*$.
\end{proposition} 
\begin{proof}
We prove just the equality $\wsetB{T}{=}{1} = \B^*\scrD \Gamma^*$
(a proof of the other equality is similar).
By induction on the length of $\alpha\in\Gamma^*$,
we show that $\alpha\in\wsetB{T}{=}{1}$ iff $\alpha\in\B^*\scrD\Gamma^*$,
using the notation $\sigma[{-}\alpha]$ and $\sigma[{+}\alpha]$
that was introduced in the proof of 
Proposition~\ref{prop-greater-zero-regular}.
For $\alpha=\varepsilon$, both sides of the equivalence are false.
Now assume that the equivalence holds for all configurations of
length $k$, and consider an arbitrary $X\alpha\in\Gamma^+$ where
$|\alpha|=k$. If $X\alpha\in\wsetB{T}{=}{1}$, we distinguish two 
possibilities: 
\begin{itemize}

\item There is a strategy $\sigma\in\Sigma$ such that for all $\pi\in\Pi$,
the probability of reaching $T$ from $X\alpha$ without prior reaching 
$\alpha$ is $1$ in $G_\Delta(\sigma,\pi)$.
Then $\sigma[{-}\alpha]$ is $(T,{=}1)$-winning in $X$, 
which means that $X\in\wsetB{T}{=}{1}$, i.e., $X\in\scrD$.

\item There is a strategy $\sigma\in\Sigma$ such that for all $\pi\in\Pi$,
the probability of reaching $T$ from $X\alpha$ in the 
play $G_\Delta(\sigma,\pi)$ is $1$, but for some $\hat{\pi} \in \Pi$, the
configuration $\alpha$ is reached with a positive probability before
reaching $T$. In this case, consider again the strategy 
$\sigma[{-}\alpha]$, which is 
\mbox{$(T_\varepsilon,{=}1)$-winning} in $X$ and hence $X\in\B$.
Moreover, observe that the strategy $\sigma$ is $(T,{=}1)$-winning in 
$\alpha$.
Hence, $\alpha\in\wsetB{T}{=}{1}$ and by applying induction hypothesis
we obtain $\alpha\in\B^*\scrD\Gamma^*$.
\end{itemize}

For the opposite direction, we assume $X\alpha\in\B^*\scrD\Gamma^*$, and 
distinguish the following possibilities:
\begin{itemize}
\item $X\in\scrD$ and there is a $(T,{=}1)$-winning strategy
  $\sigma\in\Sigma$ in $X$. Then  $\sigma[{+}\alpha]$ 
  is $(T,{=}1)$-winning in $X\alpha$. Thus,
  $X\alpha\in\wsetB{T}{=}{1}$.
\item  $X\in\B$ and $\alpha\in\B^*\scrD\Gamma^*$.
  Then there is a
$(T_\varepsilon,{=}1)$-winning strategy $\sigma_1\in\Sigma$ in $X$.
By applying induction hypothesis, there is a
$(T,{=}1)$-winning strategy $\sigma_2\in\Sigma$ in $\alpha$.
Now we can set up a $(T,{=}1)$-winning strategy in $X\alpha$,
which behaves like $\sigma_1[{+}\alpha]$ until $\alpha$ is reached,
and from that point on it behaves like $\sigma_2$.
Hence, $X\alpha\in\wsetB{T}{=}{1}$.
\qed
\end{itemize}
\renewcommand{\qed}{}
\end{proof}

\label{page-discussion-B-C}
By Theorem~\ref{thm-determinacy}, $\B = \Gamma \smallsetminus \A$
and $\scrD = \Gamma \smallsetminus \C$. Hence, it suffices to compute
the sets $\A$ and $\C$. Further, observe that if the set $\A$ 
is computable for an arbitrary stochastic BPA game, then the
set $\C$ is also computable with the same complexity.
This is because
$X \in \wsetD{T}{<}{1}$ iff 
$\tilde{X} \in \wsetD{\tilde{T}_\varepsilon}{<}{1}$, where
$\wsetD{\tilde{T}_\varepsilon}{<}{1}$ is considered in 
a stochastic BPA
game $\tilde{\Delta}$  obtained from $\Delta$ by adding 
two fresh symbols $\tilde{X}$ and $Z$ to $\Gamma_\bigcirc$
together with the rules 
$\tilde{X} \btran{1} XZ$, $Z \btran{1} Z$, and setting 
$\tilde{T} = T$.
Hence, the core of the whole problem is to design an algorithm which
computes the set $\A$. 

In the next definition we introduce the crucial notion of a 
\emph{terminal} set of stack symbols, which plays a key role
in our considerations. 

\begin{definition}
  A set $M \subseteq \Gamma$ is \emph{terminal} if the following
  conditions are satisfied:
  \begin{itemize}
  \item $\Gamma_T \cap M = \emptyset$;
  \item for every $Z \in M \cap (\Gamma_\Box \cup \Gamma_\bigcirc)$
     and every rule of the form $Z \btran{} \alpha$ we have that
     $\alpha \in M^*$;
  \item for every $Z \in M \cap \Gamma_\Diamond$
     there is a rule $Z \btran{} \alpha$ such that $\alpha \in M^*$.
  \end{itemize}
\end{definition}

Since the empty set is terminal and the union of two terminal sets
is terminal, there is the greatest terminal set that will be
denoted by~$C$ in the rest of this section. Also note that $C$
determines a stochastic BPA game $\Delta_C$ obtained from 
$\Delta$ by restricting the set of stack symbols to $C$ and 
including all rules $X \btran{} \alpha$ where $X,\alpha \in C^*$.
The set of rules of $\Delta_C$ is denoted by $\btran[C]{}$.
The probability of stochastic rules in $\Delta_C$ is the same as 
in~$\Delta$.

\begin{definition}
\label{def-witness}
  A stack symbol $Y \in \Gamma$ is a \emph{witness} if one of the
  following conditions is satisfied:
  \begin{itemize}
  \item[(1)] $Y \in \wsetD{T_\varepsilon}{=}{0}$;
  \item[(2)] $Y \in C$ and $Y \in \wsetD{\varepsilon}{<}{1}$, where
     the set $\wsetD{\varepsilon}{<}{1}$ is computed in $\Delta_C$.
  \end{itemize}
  The set of all witnesses is denoted by $W$.
\end{definition}

In the next lemma we show that every witness belongs to the set~$\A$.

\begin{lemma}
\label{lem-witness-strategy}
  The problem whether $Y \in W$ for a given $Y \in \Gamma$ is in 
  $\NP \cap \coNP$. Further, there is a SMD strategy $\pi \in \Pi$
  constructible by a deterministic polynomial-time algorithm 
  with $\NP \cap \coNP$
  oracle such that for all $Y \in W$ and $\sigma \in \Sigma$ we
  have that $\calP_{Y}^{\sigma,\pi}(\R(T_\varepsilon)) < 1$.
\end{lemma}
\begin{proof}
  Let $W_2$ be the set of all type~(2) witnesses of~$\Delta$, and let
  $W_1$ be the set of all type~(1) witnesses that are not type~(2) 
  witnesses (see Definition~\ref{def-witness}).

  Let us first consider the BPA game $\Delta_C$ (note that $\Delta_C$
  is constructible in polynomial time). By the results of 
  \cite{EY:RMDP-efficient}, there are SMD strategies $\sigma'$ and $\pi'$ 
  in $G(\Delta_C)$ such that $\sigma'$ 
  is $(\varepsilon,{=}1)$-winning in every configuration of
  $\wsetB{\varepsilon}{=}{1}$ and $\pi'$ is $(\varepsilon,{<}1)$-winning
  in every configuration of $\wsetD{\varepsilon}{<}{1}$ (here
  the sets $\wsetB{\varepsilon}{=}{1}$ and $\wsetD{\varepsilon}{<}{1}$
  are considered in $\Delta_C$). In \cite{EY:RMDP-efficient}, it is also 
  shown that the problem whether a given SMD strategy is
  $(\varepsilon,{=}1)$-winning (or $(\varepsilon,{<}1)$-winning)
  in every configuration of $\wsetB{\varepsilon}{=}{1}$
  (or $\wsetD{\varepsilon}{<}{1}$) is decidable in polynomial time.
  Hence, the problem whether a given $Y \in \Gamma$ belongs to $W_2$
  is in $\NP \cap \coNP$,
  and the strategy $\pi'$ is constructible by an algorithm which
  successively fixes one of the available rules for every  
  $Y \in \Gamma_\Diamond \cap C$ so that the set $\wsetD{\varepsilon}{<}{1}$
  remains unchanged when all of the other rules with $Y$ on the left-hand
  side are removed from $\Delta_C$. Obviously, this algorithm needs 
  only $\calO(|\Delta_C|)$ time to fix such a rule for every 
  $Y \in \Gamma_\Diamond \cap C$ (i.e., to construct the strategy $\pi'$)
  if it is equipped with a $\NP \cap \coNP$ oracle which can be used to
  verify that the currently considered rule is a correct one.

  The strategy $\pi'$ can also be applied in the game $G(\Delta)$
  (for every $Z \in \Gamma_\Diamond \smallsetminus C$ we just define
  $\pi'(Z)$ arbitrarily). Since $\Gamma_T \cap C = \emptyset$, 
  for all $Y \in W_2$ and $\sigma \in \Sigma$
  we have that $\calP_{Y}^{\sigma,\pi'}(\R(T_\varepsilon)) < 1$.
  
  The remaining witnesses of $W_1$ can be discovered in polynomial time, and 
  there is a SMD strategy $\pi'' \in \Pi$ constructible 
  in polynomial time such that for all $Y \in W_1$ and
  $\sigma \in \Sigma$ we have that
  $\calP_{Y}^{\sigma,\pi''}(\R(T_\varepsilon)) = 0$ or
  $\calP_{Y}^{\sigma,\pi''}(\R(W_2\Gamma^*)) > 0$. This follows
  directly from Theorem~\ref{thm-BPA-case-zero} and 
  Remark~\ref{rem-extend-case-zero}. 

  The strategy $\pi$ is constructed simply by ``combining'' the
  strategies $\pi'$ and $\pi''$. That is, $\pi$ behaves like $\pi'$
  (or $\pi''$) in all configurations $Y\alpha$ where $Y \in W_2$ 
  (or $Y \in W_1$).
\end{proof}

Due to Lemma~\ref{lem-witness-strategy}, we have that $W \subseteq \A$. 
One may be tempted
to think that the set $\A$ is just the \emph{attractor} of
$W$, denoted $\Att(W)$, which consists of all stack symbols from which 
player~$\Diamond$ can enforce visiting a witness with a positive 
probability. However, this (natural) hypothesis
is false, as demonstrated by the following example:

\begin{example}
\label{exa-no-reach}
Consider a stochastic BPA game 
$\hat{\Delta} = (\{X,Y,Z,R\},\btran{},(\{X\},\emptyset,
\{Y,Z,R\}),\Prob)$, where $X \btran{} X$, $X \btran{} Y$, $X \btran{} Z$,
$Y \btran{1} Y$, $Z \btran{1/2} Y$, $Z \btran{1/2} R$, $R \btran{1} R$,
and the set $T_\Gamma$ contains just $R$. The game is initiated in $X$,
and the relevant part of $G_{\hat{\Delta}}$
(reachable from $X$) is shown in the following figure:
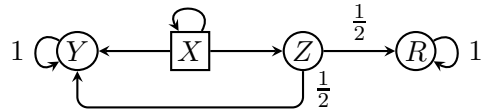
\begin{center}
\begin{tikzpicture}[x=1.5cm,y=1.5cm,font=\small]
\node (Y) at (0,0)   [ran] {$Y$};
\node (X) at (1,0)   [max] {$X$};
\node (Z) at (2,0)   [ran] {$Z$};
\node (R) at (3,0)   [ran] {$R$};
\draw [transition] (Y) to [loop left]  node[left]  {$1$} (Y); 
\draw [transition] (X) to [loop above] (X); 
\draw [transition] (X) to (Y); 
\draw [transition] (X) to (Z); 
\draw [transition] (Z) to node[above] {$\frac{1}{2}$} (R); 
\draw [transition] (R) to [loop right] node[right]  {$1$} (R); 
\draw [transition, rounded corners] (Z) -- node[right] {$\frac{1}{2}$} +(0,-.5) 
  -- +(-2,-.5) -- (Y); 
\end{tikzpicture}
\end{center}
Observe that $\A = \{X,Y,Z\}$, $C = W = \{Y\}$, but $\Att(\{Y\}) = \{Z,Y\}$.
\end{example}
The problem is that, in general, 
player~$\Box$ cannot be ``forced'' to enter $\Att(W)$ (in 
Example~\ref{exa-no-reach}, player~$\Box$ can always select the rule
$X \btran{} X$ and thus avoid entering $\Att(\{Y\})$). Nevertheless,
observe that player~$\Box$ has essentially only two options: 
she either enters a symbol of $\Att(W)$, or avoids visiting the symbols
of $\Att(W)$ completely. The second possibility is analyzed by
``cutting off'' the set $\Att(W)$ from the considered BPA game,
and  recomputing the set of all witnesses together with its attractor 
in the resulting BPA game which is smaller than the original one.
In Example~\ref{exa-no-reach}, we ``cut off'' the attractor
$\Att(\{Y\})$ and thus obtain a smaller BPA game with just one
symbol $X$ and the rule $X \btran{} X$. Since that $X$ is a witness 
in this game, it can be safely added to the set $\A$.
In general, the algorithm for computing the set $\A$
proceeds by putting $\A := \emptyset$ and then repeatedly computing 
the set $\Att(W)$, setting $\A :=  \A \cup \Att(W)$, and ``cutting off''
the set $\Att(W)$ from the game. This goes on until the set $\Att(W)$ 
becomes empty. 

We start by demonstrating that if $\A \neq \emptyset$ then there is at 
least one witness.  This is an important (and highly non-trivial) result,
whose proof is postponed to Section~\ref{sec:witness-proof}.

\begin{proposition}
  If $\A \neq \emptyset$, then $W \neq \emptyset$. 
\label{prop:witnesses}
\end{proposition}

In other words, the non-emptiness of $\A$ is always certified by at least
one witness, and hence each stochastic BPA game with a non-empty $\A$
can be made smaller by ``cutting off'' $\Att(W)$. The procedure which ``cuts
off'' the symbols $\Att(W)$ is not completely trivial. A naive idea of
removing the symbols of $\Att(W)$ together with the rules where they appear
(this was used for the stochastic BPA game of Example~\ref{exa-no-reach})
does not always work. This is illustrated in the following example:

\begin{example}
\label{ex:cut}
  Consider a stochastic BPA game 
  $\hat{\Delta} = (\{X,Y,Z,R\},\btran{},(\{X\},\emptyset,\{Y,Z,R\}),\Prob)$,
  where 
  \[
  X \btran{} X,\ X \btran{} Y,\ X \btran{} ZY,\
  Y \btran{1} Y,\ Z \btran{1/2} X,\ Z \btran{1/2} R,\ R \btran{1} R
  \]
  and ${\hat{\Gamma}}_T =  \{R\}$. The game is initiated in $X$ 
  (see Fig.~\ref{fig:ex_cut}).
  We have that $\A = \{Y\}$ (observe that 
  $X,Z,R \in \wsetB{T_\varepsilon}{=}{1}$, because the strategy $\sigma$ 
  of player~$\Box$ which always selects the rule  $X \btran{} ZY$
  is $(T,{=}1)$-winning). Further, we have that 
  $C = W = \Att(W) = \{Y\}$. If we remove
  $Y$ together with all rules where $Y$ appears, we obtain the game
  $\Delta' = (\{X,Z,R\},\btran{},(\{X\},\emptyset,\{Z,R\}),\Prob)$,
  where $X \btran{} X$, $Z \btran{1/2} X$, $Z \btran{1/2} R$, $R \btran{1} R$.
  In the game $\Delta'$, $X$ becomes a witness and hence the
  algorithm would incorrectly put $X$ into $\A$.
\end{example}

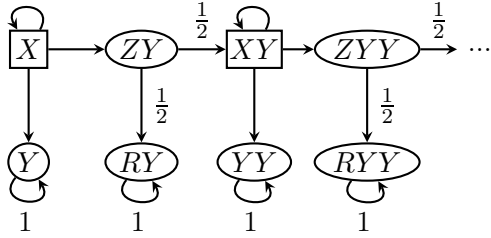
\begin{figure}[ht]
  \centering
\begin{tikzpicture}[x=1.5cm,y=1.5cm,font=\small]
\node (X)   at (0,1) [max] {$X$};
\node (ZY)  at (1,1) [ran] {$ZY$};
\node (XY)  at (2,1) [max] {$XY$};
\node (ZYY) at (3,1) [ran] {$ZYY$};
\node (Y)   at (0,0) [ran] {$Y$};
\node (RY)  at (1,0) [ran] {$RY$};
\node (YY)  at (2,0) [ran] {$YY$};
\node (RYY) at (3,0) [ran] {$RYY$};
\node (etc) at (4,1) [] {...};
\draw [transition] (X) to [loop above] (X); 
\draw [transition] (XY) to [loop above] (XY); 
\draw [transition] (Y) to [loop below] node[below] {1} (Y); 
\draw [transition] (RY) to [loop below] node[below] {1} (RY); 
\draw [transition] (YY) to [loop below] node[below] {1} (YY); 
\draw [transition] (RYY) to [loop below] node[below] {1} (RYY); 
\draw [transition] (X) to (ZY); 
\draw [transition] (X) to (Y); 
\draw [transition] (ZY) to node[above] {$\frac{1}{2}$} (XY); 
\draw [transition] (ZY) to node[right] {$\frac{1}{2}$} (RY); 
\draw [transition] (XY) to (ZYY); 
\draw [transition] (XY) to (YY); 
\draw [transition] (ZYY) to node[right] {$\frac{1}{2}$} (RYY); 
\draw [transition] (ZYY) to node[above] {$\frac{1}{2}$} (etc);
\end{tikzpicture}
  \caption{The game of Example~\ref{ex:cut}}
  \label{fig:ex_cut}
\end{figure}

Hence, the ``cutting'' procedure must be designed more carefully. 
Intuitively, we do not remove rules of the form $X \btran{} ZY$,
where $Y \in \Att(W)$, but change them into  $X \btran{} \pp{Z}$, where
the plays initiated in $\pp{Z}$ ``behave'' like the ones initiated in $Z$ 
with the exception that $\varepsilon$ cannot be reached whatever the 
players do.

Now we show how to compute the set $\A$, formalizing the intuition
given above. To simplify the proofs of our claims, we adopt some 
additional (safe) assumptions about the considered BPA game $\Delta$.

\begin{definition}\label{ass:rules}
We say that $\Delta$ is in \emph{special normal form (SNF)}
if all of the following conditions are satisfied:
\begin{itemize}
\item For every $R \in \Gamma_T$ we have that $R \in \Gamma_\bigcirc$
  and $R \btran{1} R$.
\item For every rule $X \btran{} \alpha$ where 
  $X \in \Gamma_\Diamond \cup \Gamma_\bigcirc$ we have that $\alpha \in \Gamma$.
\item The set $\Gamma_\Box$ can be partitioned into three disjoint
  subsets $\Gamma[1]$,  $\Gamma[2]$, and $\Gamma[3]$ so that
  \begin{itemize}
  \item if $X \in \Gamma[1]$ and $X \btran{} \alpha$, then 
     $\alpha \in \Gamma$;
  \item if $X \in \Gamma[2]$, then $X \btran{} \varepsilon$ and
     there is no other rule of the form $X \btran{} \alpha$;
  \item if $X \in \Gamma[3]$, then $X \btran{} YZ$ for some 
     $Y,Z \in \Gamma$, and there is no other rule of the form 
     $X \btran{} \alpha$.
  \end{itemize}
\end{itemize}
\end{definition}

Note that every BPA game can be efficiently transformed into 
an ``equivalent'' BPA game in SNF
by introducing fresh stack symbols (which belong to player~$\Box$) and
adding the corresponding dummy rules. For example, if the original BPA game
contains the rules $X \btran{} \varepsilon$ and $X \btran{} YZ$,
then the newly constructed BPA game in SNF contains the rules
$X \btran{} E$, $X \btran{} P$, $E \btran{} \varepsilon$, 
$P \btran{} YZ$, where $E,P$ are fresh stack symbols that belong to
player~$\Box$. Obviously, the set $\A$ of the original BPA game
is the set $\A$ of the newly constructed BPA game restricted 
to the stack symbols of the original BPA game.

So, from now on we assume that the considered BPA game $\Delta$ is in SNF.
In particular, note that only player~$\Box$ can change the height of 
the stack; and if she can do it, then she cannot do anything else 
for the given stack symbol. 

Our algorithm for computing the set $\A$ consists of two parts,
the procedure \FuncSty{Init} and the procedure \FuncSty{Main}.  
The procedure \FuncSty{Init} transforms the BPA game $\Delta$ into another 
BPA game $\initdelta$, which is then used as an input for the procedure
\FuncSty{Main} which computes the set $\A$ of $\initdelta$.

For every $X \in \Gamma$, let $\pp{X}$ be a fresh ``twin'' of $X$, and
let $\pp{\Gamma} = \{\pp{X} \mid X \in \Gamma\}$. Similarly, for every
$\odot \in \{\bigcirc,\Diamond,\Box\}$ we put
$\pp{\Gamma}_{\odot}=\{\pp{X}\mid X\in \Gamma_{\odot}\}$.
The procedure \FuncSty{Init} inputs the BPA game $\Delta$
and outputs another BPA game 
$\initdelta = (\bar{\Gamma},\btran{},
(\bar{\Gamma}_\Box,\bar{\Gamma}_\Diamond,\bar{\Gamma}_\bigcirc),\Prob)$ 
where $\bar{\Gamma} = \Gamma \cup \pp{\Gamma}$, 
$\bar{\Gamma}_{\odot} = \Gamma_{\odot} \cup  \pp{\Gamma}_{\odot}$ for 
every $\odot \in \{\bigcirc,\Diamond,\Box\}$, and the rules are constructed
as follows: 
\begin{itemize}
\item if $X \btran{} \varepsilon$ is a rule of $\Delta$, then 
  $X \btran{} \varepsilon$ and $\pp{X} \btran{} \pp{X}$ are rules
  of $\initdelta$;
\item if $X \btran{} Y$ is a rule of $\Delta$, then 
  $X \btran{} Y$ and $\pp{X} \btran{} \pp{Y}$ are rules
  of $\initdelta$;
\item if $X \btran{} YZ$ is a rule of $\Delta$, then 
  $X \btran{} YZ$ and $\pp{X} \btran{} Y\pp{Z}$ are rules
  of $\initdelta$;
\item $\initdelta$ has no other rules.
\end{itemize}
Further, if $X \btran{x} Y$ in $\Delta$, then $X \btran{x} Y$ and
$\pp{X} \btran{x} \pp{Y}$ in $\initdelta$. We put
$\bar{\Gamma}_T = \{R,\pp{R} \mid R \in \Gamma_T\}$.

Intuitively, the only difference between $X$ and $\pp{X}$ is that
$\pp{X}$ can never be fully removed from the stack. Also observe that 
the newly added stack symbols of $\pp{\Gamma}$ are unreachable from
the original stack symbols of $\Gamma$. Hence, the set $\A$ of 
$\Delta$ is obtained simply by restricting the set $\A$ of $\initdelta$ 
to the symbols of $\Gamma$. In the rest of this section, we adopt 
the following convention:
the elements of $\Gamma$ are denoted by $X,Y,Z,\ldots$, the corresponding
elements of $\pp{\Gamma}$ are denoted by $\pp{X},\pp{Y},\pp{Z},\ldots$, 
and for every $X \in \Gamma$, the symbol $\bar{X}$ denotes either $X$ 
or $\pp{X}$. 

The set $\A$ of $\initdelta$ is computed by the procedure~\FuncSty{Main}
(see page~\pageref{proc-main}). At line~\ref{alg:att-compute}, we assign 
to $\M$ the least fixed-point of the function 
$\FuncSty{Att}_{\Theta,W} : 2^{\bar{\Gamma}} \rightarrow 2^{\bar{\Gamma}}$,
where $\Theta$ is an auxiliary BPA game maintained by the 
procedure~\FuncSty{Main} and $W$ is a subset of stack symbols 
of~$\Theta$. The function $\FuncSty{Att}_{\Theta,W}$ is defined as follows
(the set of rules of $\Theta$ is denoted by $\bbtran{}$):
  \begin{eqnarray*}
   \FuncSty{Att}_{\Theta,W}(S) & = & W\\
   & \cup & \{\bar{A} \in \bar{\Gamma}_\bigcirc \cup \bar{\Gamma}_\Diamond 
       \mid \mbox{there is a rule } \bar{A} \bbtran{} \bar{B} \mbox{ where }
       \bar{B} \in S\}\\
   & \cup & \{\bar{A} \in \bar{\Gamma}[1]
       \mid \bar{B} \in S 
          \mbox{ for all } \bar{A} \bbtran{} \bar{B} \}\\
   & \cup & \{\bar{A} \in \bar{\Gamma}[3]
       \mid \bar{A} \bbtran{} Y\bar{C} \mbox{ where }
       Y \in S \mbox{ or } \pp{Y},\bar{C} \in S \}
  \end{eqnarray*}
Note that the procedure~\FuncSty{Main} actually computes
the sets $\A$ and $\C$ of $\initdelta$ simultaneously, as stated
in the following proposition. A proof is postponed to 
Section~\ref{sec-correctness}. 

\begin{procedure}
\caption{Main()}
\label{proc-main}
\KwData{A BPA game $\initdelta = (\bar{\Gamma},\btran{},
   (\bar{\Gamma}_\Box,\bar{\Gamma}_\Diamond,\bar{\Gamma}_\bigcirc),\Prob)$.}
\KwResult{The sets $\mathcal{W}$ and $\U$.}
\SetNoFillComment

$\mathcal{W}\coloneqq \emptyset$;\ \ $\U\coloneqq \emptyset$;\ 
$\Theta \coloneqq \bar{\Delta}$\;

\While{the greatest set $W$ of witnesses in $\Theta$ is not empty}
{\nllabel{alg:comp-wit}
  $\M \coloneqq \mbox{the least fixed-point of } \FuncSty{Att}_{\Theta,W}$\;\nllabel{alg:att-compute}

  \For{every  $\bar{A} \in \M$}
    {\nllabel{alg:symbol-remove}
     remove the symbol $\bar{A}$ and all rules with $\bar{A}$
     on the left-hand side\;
    } 

  \For{every rule $\bar{A} \bbtran{} \bar{B}$ where
       $\bar{A} \in \bar{\Gamma}_{\Box} \smallsetminus \M$ and
       $\bar{B} \in \M$}
    {%
     remove the rule $\bar{A} \bbtran{} \bar{B}$\;
    }\nllabel{alg:intdc_add}  %
    
  \For{every rule $\bar{A} \bbtran{} Y\bar{C}$ where
       $\bar{A} \in \bar{\Gamma}_{\Box} \smallsetminus \M$ and
       $\bar{C} \in \M$} 
  {\nllabel{alg:W-under}
     replace the rule $\bar{A} \bbtran{} Y\bar{C}$ with the rule
     $\bar{A} \bbtran{} \pp{Y}$\; 
     \nllabel{alg:remove-finish}
  }
  $\mathcal{W} \coloneqq \mathcal{W}\cup \M$\;\nllabel{alg:add-cal-W}
  $\U \coloneqq \U \cup \{\bar{Y} \mid \pp{Y} \in \mathcal{W}\}$\;
  \nllabel{alg:add-cal-U}
}
\Return{$\mathcal{W},\U$}\;
\end{procedure}

\begin{proposition} 
\label{prop:cor_alg}
The sets $\mathcal{W}$ and $\U$ computed by the procedure~\FuncSty{Main} are
exactly the sets $\A$ and $\C$ of the BPA game $\initdelta$, respectively.
\end{proposition}

Now, let us analyze the complexity of the procedure~\FuncSty{Main}.
Obviously, the main loop initiated at line~\ref{alg:comp-wit}
terminates after $\calO(|\initdelta|)$ iterations. In each iteration,
we need to compute the greatest set of witnesses $W$ of the current game,
which is the only step that needs exponential time. Hence, the running
time of the procedure~\FuncSty{Main} is \emph{exponential} in the
size of $\initdelta$. Nevertheless, the procedure~\FuncSty{Main} can
be easily modified into its \emph{non-deterministic} variant
\FuncSty{Main-NonDet} where every computation terminates
after a polynomial number of steps, and all ``successful'' computations
of \FuncSty{Main-NonDet} output the same sets $\W,\U$ as the 
procedure~\FuncSty{Main}. This means that the membership problem as well as 
the non-membership problem for the set~$\A$ is in $\NP$, which implies 
that both problems are in fact in $\NP \cap \coNP$. The same applies
to the set $\C$. The only difference
between the  procedures~\FuncSty{Main} and \FuncSty{Main-NonDet} is the
way of computing the greatest set of witnesses $W$. 
Due to Lemma~\ref{lem-witness-strategy}, the problem whether
$Y \in W$ for a given $Y \in \Gamma$ is in $\NP \cap \coNP$. 
Hence, the membership as well as the non-membership to $W$ is certified 
by certificates of polynomial size that are verifiable in polynomial time
(in the proof of Lemma~\ref{lem-witness-strategy}, we indicated 
how to construct these certificates, but this is not important now). 
The procedure \FuncSty{Main-NonDet} guesses the set $W$
together with a tuple of certificates that are supposed to prove that
the guess was fully correct (i.e., the guessed set is exactly the set of
all witnesses).
Then, all of these certificates are verified. If some of them turns out
to be invalid, the procedure \FuncSty{Main-NonDet} terminates immediately
(this type of termination is considered ``unsuccessful'').
Otherwise, the procedure \FuncSty{Main-NonDet} proceeds by performing 
the same instructions as the procedure~\FuncSty{Main}.

Since the membership problem for the sets $\A,\C$ is in $\NP \cap \coNP$,
the membership problem for the sets $\B,\scrD$ is also in
$\NP \cap \coNP$ (see the discussion at page~\pageref{page-discussion-B-C}).
Hence, an immediate consequence of the previous observations and 
Proposition~\ref{prop-equal-one-regular} is the following:

\begin{theorem}
\label{thm-BPA-case-one}
  The membership to $\wsetB{T}{=}{1}$ and $\wsetD{T}{<}{1}$ is 
  in $\NP \cap \coNP$. Both sets are effectively regular, and 
  the associated finite-state automata are constructible by a deterministic
  polynomial-time algorithm with $\NP \cap \coNP$ oracle.
\end{theorem}

Since the arguments used in the proof of 
Proposition~\ref{prop:cor_alg} are mostly constructive, 
the winning strategies for both players are effectively regular.
This is stated in our final theorem (a proof can be found
in Section~\ref{sec-correctness}).

\begin{theorem}\label{thm:box_strat_1}
  There are regular strategies $\sigma \in \Sigma$ and $\pi \in \Pi$
  such that $\sigma$ is $(T,{=}1)$-winning in every 
  configuration of $\wsetB{T}{=}{1}$ and $\pi$ is
  $(T,{<}1)$-winning in every configuration of 
  $\wsetD{T}{<}{1}$. 
  Moreover, the strategies $\sigma$ and $\pi$ are constructible 
  by a deterministic polynomial-time algorithm with $\NP \cap \coNP$ 
  oracle.
\end{theorem}

\section{Proofs of Section~\ref{sec:one}}
\label{sec:proofs}

In this section we present the proofs that were omitted in
Section~\ref{sec:one}.

\subsection{A Proof of Proposition~\ref{prop:witnesses}}
\label{sec:witness-proof}

We start by formulating a simple corollary to 
Proposition~\ref{prop-greater-than-zero}, which turns out to be
useful at several places.

\begin{proposition}
\label{prop:stack-bound}
  Let $\sigma\in \Sigma$ be a strategy of player~$\Box$ which always
  returns a uniform probability distribution over the available
  outgoing edges. Then for every $X\in \wsetB{T}{>}{0} \cap \Gamma$ 
  (or $X\in \wsetB{T_\eps}{>}{0} \cap \Gamma$)
  and every $\pi\in\Pi$ there
  is a path $w$ from $X$ to $T$ (to $T_\eps$, resp.)
  in $G_\Delta(\sigma, \pi)$ such that
  \begin{enumerate}
  \item the length of $w$ is at most $2^{2|\Gamma|}$;
  \item the length of all configurations visited by $w$ 
     is at most $2|\Gamma|$.
  \end{enumerate}
\end{proposition}
\begin{proof}
Let us consider the sets $A_i$ and $B_i$ from the proof of 
Proposition~\ref{prop-greater-than-zero}. Recall that
$\wsetB{T}{>}{0}\cap\Gamma=\bigcup_{i=0}^{2|\Gamma|} A_i$ and
$\wsetB{T_\eps}{>}{0}\cap\Gamma=\bigcup_{i=0}^{2|\Gamma|} B_i$.
By induction on~$i$, we prove that for every
$X\in A_i$ (or $X\in B_i$) and every $\pi\in\Pi$ there
is a path $w$ from $X$ to $T$ (or to $T_\eps$, resp.)
in $G_\Delta(\sigma, \pi)$ such that
  \begin{enumerate}
  \item[(1)] the length of $w$ is at most $2^i$;
  \item[(2)] the length of all configurations visited by $w$ 
  is at most $i$.
  \end{enumerate}
The case $i=1$ is trivial, as $\A_1=\B_1=\Gamma_T$. Now assume that $i>1$. If 
$X\in A_i\cap (\Gamma_{\Box}\cup \Gamma_{\bigcirc})$, then
by the definition of $A_i$, there is a transition $X\btran{} \gamma$ such that
$\gamma\in \Gamma_T\cup A_{i-1} \Gamma\cup B_{i-1} A_{i-1}\cup A_{i-1}$.
By induction hypothesis, there is a path $w'$ from $\gamma$ to $T$ in 
$G_\Delta(\sigma,\pi)$
of length at most $2^{i}+2^{i}=2^{i+1}$ such that the length of all
configurations entered by $w'$ 
is at most $\max\{i{+}1,i\}=i+1$.
The rest follows from the fact that $\sigma$ always returns a uniform
probability distribution, and if 
$X \in A_i\cap \Gamma_{\Diamond}$, then all outgoing transitions of $X$
have the form $X\btran{} \gamma$ where
$\gamma\in \Gamma_T\cup A_{i-1} \Gamma\cup B_{i-1} A_{i-1}\cup A_{i-1}$
(we use induction hypothesis to obtain the desired result).
The case when $X\in B_i$ follows similarly.
\end{proof}

Proposition~\ref{prop:witnesses} is obtained as a corollary to the
following (stronger) claim that will also be used later when 
synthesizing a regular $(T,{=}1)$-winning strategy for player~$\Box$.

\begin{proposition}
\label{prop:no-wit-reg-win}
  Let $W$ be the set of all witnesses (see
  Definition~\ref{def-witness}). If $W=\emptyset$, then there is a
  regular strategy $\sigma$ of player~$\Box$, computable
  in polynomial time, which is $(T_{\varepsilon},{=}1)$-winning in
  every configuration of $\Delta$.
\end{proposition}
  In particular, if $W=\emptyset$ then $\A = \emptyset$, and thus
  we obtain Proposition~\ref{prop:witnesses}. Now we prove 
  Proposition~\ref{prop:no-wit-reg-win}, relaying on further
  technical observations that are formulated and proved at appropriate places.
    
  As $W = \emptyset$, the two conditions of
  Definition~\ref{def-witness} are not satisfied by any 
  $Y \in \Gamma$. This means that for all $Y \in C$ we have that
  $Y \in \wsetB{\eps}{=}{1}$, where the set $\wsetB{\eps}{=}{1}$
  is computed in $\Delta_C$ (we again use Theorem~\ref{thm-determinacy}). 
  Due to \cite{EY:RMC-RMDP}, there exists a SMD
  strategy $\sigma_T$ for player~$\Box$ in $G_{\Delta_C}$ such that
  for every $Y \in C$ and every strategy $\pi$ of player~$\Diamond$ 
  in $G_{\Delta_C}$ we have that $\calP_Y^{\sigma_T,\pi}(\R(\varepsilon)) = 1$. 

  Let $\sigma_{U}$ be the SMD strategy of player~$\Box$ 
  which always returns the uniform probability distribution over
  the available edges.
  In the proof we use the following simple property of $\sigma_U$, which
  follows easily from Proposition~\ref{prop:stack-bound}.

\begin{lemma}\label{lem:wit-1}
  There is $\xi>0$ such that for every $X\in\Gamma$ and every
  $\pi\in\Pi$ there is a path $w$ from $X$ to 
  a configuration of $T_\varepsilon$ in
  $G_\Delta(\sigma_U,\pi)$ satisfying the following: The length of all
  configurations visited by $w$ is bounded by $2|\Gamma|$, and the 
  probability of $\run(w)$ in
  $G_\Delta(\sigma_U,\pi)$ is at least $\xi$. 
\end{lemma}
\begin{proof}
Since $W=\emptyset$, there are no type~(1) witnesses (see 
Definition~\ref{def-witness}), i.e., 
$\Gamma \cap \wsetD{T_\varepsilon}{=}{0} = \emptyset$, which means that
$\Gamma\subseteq\wsetB{T_\eps}{>}{0}$ by Theorem~\ref{thm-determinacy}.
Let  $\pi\in\Pi$ be an arbitrary
(possibly randomized) strategy. We define the associated deterministic
strategy $\hat{\pi}$, which for every finite sequence of configurations
$\alpha_1,\ldots,\alpha_n$ selects an edge $\alpha_n \gtran{} \beta$
such that $\alpha_n \gtran{} \beta$ is assigned a maximal probability
in the distribution assigned to $\alpha_1,\ldots,\alpha_n$ by 
the strategy $\pi$.
In other words, $\alpha_n \gtran{} \beta$ is an edge selected
with a maximal probability by $\pi$. If there are several candidates 
for $\alpha_n \gtran{} \beta$, any of them can be chosen.
Obviously, every path in $G_\Delta(\sigma_U,\hat{\pi})$ initiated
in $X$ is also a path in $G_\Delta(\sigma_U,\pi)$ initiated
in $X$. Due to Proposition~\ref{prop:stack-bound}, 
there is a path $\hat{w}$ from $X$ to $T_\varepsilon$
in $G_\Delta(\sigma_U,\hat{\pi})$ such that
the length of $\hat{w}$ is bounded by $2^{2|\Gamma|}$ and the stack height
of all configurations visited by $\hat{w}$ is bounded by $2|\Gamma|$. 
Now consider the corresponding path $w$ in $G_\Delta(\sigma_U,\pi)$.
The only difference between $w$ and $\hat{w}$ is that the 
probability of the transitions selected by player~$\Diamond$ 
is not necessarily one in~$w$. However, due to the definition of 
$\hat{w}$ we immediately obtain that the probability of each such 
transition is at least $\frac{1}{|{\btran{}}|}$ (this bound is not
tight but sufficient for our purposes). Since $\sigma_U$ is uniform,
the same bound is valid also for the probability of transitions
selected by player~$\Box$. Let $\mu$ be the least probability weight 
of a probabilistic rule assigned by $\Prob$.
We put 
\[
  \xi = \left(\min\{\mu,\frac{1}{|{\btran{}}|}\}\right)^{2^{2|\Gamma|}}
\]
Obviously, $\calP(\run(w)) \geq \xi$ and we are done.
\end{proof}
  
Now we are ready to define the regular strategy $\sigma \in \Sigma$ 
whose existence was promised in Proposition~\ref{prop:no-wit-reg-win}.
Recall that regular strategies are memoryless, and hence they can be 
formally understood as functions which assign to a given configuration 
$\beta$ a probability distribution on the outgoing edges of $\beta$.
For a given $X\alpha \in \Gamma_\Box\Gamma^*$, we put 
$\sigma(X\alpha) = \sigma_T(X\alpha)$ if $X\alpha$ starts with some 
$\beta \in C^*$ where $|\beta| > 2|\Gamma|$. Otherwise, we 
put $\sigma(X\alpha) = \sigma_U(X\alpha)$.
 
Observe that the strategy $\sigma$ can easily be represented by a
finite state automaton with $\calO(|\Gamma|)$ states in the sense of
Definition~\ref{def-regular-strategy}.  Moreover, such an automaton
is easily constructible in polynomial time because the set $C$ is
computable in polynomial time. So, it remains to prove that $\sigma$ is
\mbox{$(T_{\varepsilon},{=}1)$-winning} in every configuration of $\Delta$.
 
Let us fix some strategy $\pi \in \Pi$. Our goal is to show that for
every $\alpha\in \Gamma^+$ we have that 
$\calP_{\alpha}^{\sigma,\pi}(\R(T_\varepsilon)) = 1$.
Assume the converse, i.e., there is some $\alpha\in\Gamma^+$ such that 
$\calP_{\alpha}^{\sigma,\pi}(\R(T_\varepsilon))<1$. 

\emph{Proof outline:} 
Let $w$ be a run of $G_\Delta(\sigma,\pi)$. We say that given rule 
of $\Delta$ is \emph{used infinitely often in 
$w$} if the rule was used to derive infinitely many transitions of~$w$. 
Further, we say that \emph{$w$ eventually uses only a given subset 
$\leadsto$ of $\btran{}$} if there is some $i \in \Nset$ 
such that all transitions $w(j) \tran{} w(j{+}1)$, where $j \geq i$, 
were derived using a rule of $\leadsto$.

We show that the set of all runs initiated in 
$\alpha$ that do not visit $T_{\varepsilon}$ contains a subset
$V$ of positive probability such that all runs of $V$
eventually use only the rules of $\Delta_C$. Then, we show that 
player~$\Box$, who plays according to the strategy $\sigma$, selects 
the rules of $\Delta_C$ in such a way that almost all runs 
that use only the rules of $\Delta_C$ eventually terminate (i.e.,
visit the configuration~$\varepsilon$). However, this contradicts 
the fact that $V$ contains only non-terminating runs. 
Now we elaborate this outline into a formal proof.

\begin{lemma}
\label{lem:wit-3}
  There is a set of runs $V \subseteq \run(G_\Delta(\sigma,\pi),\alpha)$ 
  such that $\calP_{\alpha}^{\sigma,\pi}(V) > 0$, and for every $w \in V$
  we have that $w$ does not visit $T_\varepsilon$ and all rules that are 
  used infinitely often in $w$ belong to $\btran[C]{}$.
\end{lemma}

\begin{proof}
Let $A$ be the set of all $w \in \run(G_{\Delta}(\sigma,\pi),\alpha)$ 
such that $w$ does not visit $T_{\varepsilon}$. By our assumption, 
$\calP_{\alpha}^{\sigma,\pi}(A) > 0$. The runs of $A$ can be split
into finitely many disjoint subsets according to the set of rules
which are used infinitely often. Since $\calP_{\alpha}^{\sigma,\pi}(A) > 0$,
at least one of these subsets $V$ must have positive probability.
Let $\btran{}_V$ be the associated set of rules that are used
infinitely often in the runs of~$V$.

We prove that $\btran[V]{} \subseteq \btran[C]{}$.
Let $L \subseteq \Gamma$ be the set of all symbols that appear on the
left-hand side of some rule in $\btran[V]{}$. To show that 
$\btran[V]{} \subseteq \btran[C]{}$, it suffices to prove that
\begin{enumerate}
\item[(a)] for every $Y \in (L \smallsetminus C) \cap (\Gamma_{\bigcirc}
\cup \Gamma_{\Box})$ we have that if $Y \btran{} \beta$, then also
$Y \btran[V]{} \beta$;
\item[(b)] for all rules $Y \btran[V]{} \beta$ we have that 
  $\beta \in (L\cup C)^*$.
\end{enumerate}
Observe that (a) and (b) together imply that
$L\cup C$ is a terminal set. Hence, $L\cup C=C$ by the maximality of $C$,
and thus $\btran[V]{} \subseteq \btran[C]{}$ as needed.

Claim~(a) follows from the fact that player~$\Box$, who plays according 
to the strategy $\sigma$, selects edges uniformly at random in 
all configurations of $((L \smallsetminus C)\cap \Gamma_{\Box})\cdot \Gamma^*$. 
Then every rule $Y \btran{} \beta$, where 
$Y \in (L \smallsetminus C) \cap (\Gamma_{\bigcirc} \cup \Gamma_{\Box})$, 
has the probability of being selected greater than some fixed non-zero
constant, which means that $Y \btran[V]{} \beta$ (otherwise, the probability
of~$V$ would be zero).
 
Now we prove Claim (b). Assume that $Y \btran[V]{} \gamma$. If $\gamma=\eps$,
then $\gamma \in (L\cup C)^{*}$.  If $\gamma=P$, then surely $P\in L$ 
because configurations with $P$ on the top of the stack occur infinitely 
often in all runs of~$V$. If $\gamma=PQ$, then $P\in L$ by applying
the previous argument. If $Q \in C$, we are done. Now assume that 
$Q\notin C$.  Note that then player~$\Box$ selects edges uniformly 
at random in all configurations of the form
$\beta Q \delta$ where $|\beta|\leq 2|\Gamma|$. 
By Lemma~\ref{lem:wit-1}, there is $0<\xi<1$ such that for every
configuration of the form $PQ\delta$ there is a path $w$ from
$PQ\delta$ to $T \cup \{Q\delta\}$ in $G_\Delta(\sigma,\pi)$ satisfying
the following:
\begin{itemize}
\item all configurations in $w$ are of the form $\hat{\beta} Q\delta$ 
  where $|\hat{\beta}|\leq 2|\Gamma|$;
\item the probability of following $w$ in $G_\Delta(\sigma,\pi)$ is
  at least $\xi$.
\end{itemize}
It follows that every run of $V$ enters configurations of 
$\{Q\}\cdot\Gamma^*$ infinitely many times because every run of $V$ 
contains infinitely many occurrences of configurations of the form 
$PQ\delta$ and no run of $V$ enters $T$. Hence, $Q \in L$.
\end{proof}
Now we prove that $\calP_{\alpha}^{\sigma,\pi}(V)=0$ and obtain 
the desired contradiction. By Lemma~\ref{lem:wit-3}, all runs of
$V$ eventually use only the rules of $\btran[C]{}$. 
Each run $w \in V$ uniquely determines its shortest prefix
$v_w$ after which no rules of $\btran{} {\smallsetminus} \btran[C]{}$ are 
used and the length of each configuration visited after the prefix $v_w$ 
is at least as large as the length of the last configuration visited 
by~$v$. For a given finite path $v$ initiated in $\alpha$, let
$U_v = \{w \in V \mid v_w = v\}$. Obviously, $V = \biguplus_v U_v$.
Since there are only countably many $v$'s, it suffices to prove
that $\calP_{\alpha}^{\sigma,\pi}(U_v)=0$ for every~$v$. So, let us 
fix a finite path $v$ initiated in $\alpha$, and let $Y\beta$ be the last 
configuration visited by $v$. Intuitively, we show that after performing 
the prefix $v$, the strategies $\sigma$ and $\pi$
can be ``simulated'' by suitable strategies $\sigma'$ and $\pi'$ in the
game $G_{\Delta_C}$ so that the set of runs $U_v$ is ``projected'' 
(by ignoring the prefix $v$ and cutting off $\beta$ from the bottom 
of the stack) onto the set of runs $U$ in the play 
$G_{\Delta_C}(\sigma',\pi')$  so that
\begin{equation*}
\calP_{\alpha}^{\sigma,\pi}(U_v) \ = \ \calP_{\alpha}^{\sigma,\pi}(\run(v)) \cdot 
\calP_Y^{\sigma',\pi'}(U)
\end{equation*}
Then, we show that $\calP^{\sigma',\pi'}_Y (U) = 0$.
This is because the strategy $\sigma'$ is ``sufficiently similar''
to the strategy $\sigma_T$, and hence the probability 
of visiting $\varepsilon$ in $G_{\Delta_C}(\sigma',\pi')$ is~$1$.

Now we formalize the above intuition. First, let us realize that
every probability distribution $f$ on the outgoing edges of
a BPA configuration $\alpha$ determines a unique \emph{rule distribution}
$f_r$ on the rules of the considered BPA game such that for every
$\alpha \gtran{} \alpha'$ we have that 
$f(\alpha \gtran{} \alpha') = f_r(Z \btran{} \gamma)$, where 
$Z \btran{} \gamma$ is the rule used to derive the edge 
$\alpha \gtran{} \alpha'$.

Observe that $Y\in C$ by the definition of $U_v$. Let $\sigma'$ 
be a MR strategy for player~$\Box$ in $G_{\Delta_C}$
such that for every $\gamma\in C^+$ we have that 
$\sigma'(\gamma) = \sigma(\gamma\beta)$. 
Further, let $\pi'$ be a strategy for player~$\Diamond$ in 
$G_{\Delta_C}$ such that for all $n \in \Nset$ and all
$\alpha_1,\dots,\alpha_n \in C^*$ we have that the rule distribution
of $\pi'(Y,\alpha_1,\ldots,\alpha_n)$ is the same as the
rule distribution of $\pi(v,\alpha_1\beta,\ldots,\alpha_n\beta)$.
Observe that every run $w \in U_v$ determines a unique run
$w_C \in \run(Y)$ in $G_{\Delta_C}(\sigma',\pi')$ obtained from
$w$ by first deleting the prefix $v(0),\ldots, v(|v|-2)$ and
then ``cutting off'' $\beta$ from all configurations in the
resulting run. Let $U = \{w_C \mid w \in U_v\}$. Now it is easy
to see that $\calP_{\alpha}^{\sigma,\pi}(U_v)  = 
\calP_{\alpha}^{\sigma,\pi}(\run(v)) \cdot \calP_Y^{\sigma',\pi'}(U)$.
Note that all runs of $U$ avoid visiting
$\varepsilon$. However, we show that almost all runs of
$G_{\Delta_C}(\sigma',\pi')$ reach $\varepsilon$, which implies
$\calP_Y^{\sigma',\pi'}(U)=0$ and hence also $\calP_{\alpha}^{\sigma,\pi}(U_v) = 0$.

Observe that the strategy $\sigma'$ works as follows. There is a
constant $k\leq 2|\Gamma|$ such that in every $\gamma\in C^+$, where
$|\gamma|\leq k$, player~$\Box$ selects edges uniformly at random. 
Otherwise, player~$\Box$ selects the same edges as if she was
playing according to $\sigma_T$. We show that there is $0<\xi<1$ such 
that for every $\gamma$, where $|\gamma|\leq k$, the probability of 
reaching $\varepsilon$ from $\gamma$ in $G_{\Delta_C}(\sigma',\pi')$ is 
at least~$\xi$. Note that if player~$\Box$ was playing  uniformly in 
all configurations, the existence of such a~$\xi$ would be guaranteed
by Lemma~\ref{lem:wit-1}. However, playing according to $\sigma_T$ in 
configurations whose length exceeds $k$ can only \emph{increase} 
the probability of reaching $\varepsilon$.  Now note that almost all runs of
$\run(Y)$ in $G_{\Delta_C}(\sigma',\pi')$ visit configurations of the
form $\gamma\in C^+$, where $|\gamma|\leq k$, infinitely often.
From this we obtain that almost all runs of 
$\run(Y)$ in $G_{\Delta_C}(\sigma',\pi')$ reach~$\varepsilon$.

\subsection{Proofs of Proposition~\ref{prop:cor_alg}
and Theorem~\ref{thm:box_strat_1}}
\label{sec-correctness}

The procedure \FuncSty{Main} (see page~\pageref{proc-main}) 
starts by initializing $\mathcal{W}$ and
$\U$ to $\emptyset$, and the auxiliary BPA game $\Theta$ to
$\bar{\Delta}$ (the set of rules of $\Theta$ is denoted by
$\bbtran{}$).  In the main loop initiated at line~\ref{alg:comp-wit}
we first compute the greatest set $W$ of witnesses in the current game
$\Theta$. At line~\ref{alg:att-compute}, we assign to $\M$ the least
fixed-point of the function $\FuncSty{Att}_{\Theta,W}$. The BPA game
$\Theta$ is then modified by ``cutting off'' the set $\M$ at
lines~\ref{alg:symbol-remove}--\ref{alg:remove-finish}.  Note that the
resulting BPA game is again in~SNF and it is strictly smaller than the
original~$\Theta$.  Then, the current sets $\mathcal{W}$ and $\U$ are
enlarged at lines~\ref{alg:add-cal-W},\ref{alg:add-cal-U}, and the new
(strictly smaller) game $\Theta$ is processed in the same way. This
goes on until $\mathcal{W}$ and $\U$ stabilize, which obviously
requires only $\calO(|\initdelta|)$ iterations of the main loop.

Let $K$ be the number of iterations of the main loop. For every
\mbox{$0 \leq i \leq K$}, let $\Theta_i$, $\mathcal{W}_i$, and $\U_i$
be the values of $\Theta$, $\mathcal{W}$, and $\U$ after executing
exactly $i$~iterations. Further, $W_i$ denotes the set of all
witnesses in $\Theta_i$, and $\M_i$ denotes the least fixed-point of
$\FuncSty{Att}_{\Theta_i,W_i}$. The symbols $\Sigma_i$ and $\Pi_i$
denote the set of all strategies for player~$\Box$ and
player~$\Diamond$ in $G_{\Theta_i}$, respectively.  Finally,
$\bar{\Gamma}_i$, $\bbtran[i]{}$, $\A_i$, and $\C_i$
denote the stack alphabet, the set of all rules, the set $\A$, and the
set $\C$ of $\Theta_i$, respectively. The edge relation
of~$G_{\Theta_i}$ is denoted by $\gtran[i]{}$.  Observe that $\Theta_0 =
\bar{\Delta}$, $\mathcal{W}_0 = \emptyset$, $\A_0 = \A$, $\C_0 = \C$,
$W_K = \emptyset$, and $\mathcal{W}_K,\U_K$ is the result of the
procedure~\FuncSty{Main}. Let us note that in this section,
the sets $\wsetD{T_\varepsilon}{<}{1}$ and $\wsetD{T}{<}{1}$
are always considered in the game $\bar{\Delta} = \Theta_0$.

We start by a simple observation which formalizes the relationship 
between the symbols $X$ and $\pp{X}$ in $\bar{\Delta} = \Theta_0$.
A proof is straightforward. 

\begin{lemma}
\label{lem-relationship}
  If $X \in \wsetD{T_\varepsilon}{<}{1}$, then $\pp{X} \in \wsetD{T}{<}{1}$.
\end{lemma}

Now we show that $\mathcal{W}_K = \A_0$ and $\U_K = \C_0$. 
For every $0 \leq i \leq K$, let 
$\wsetD{T_\varepsilon,i}{<}{1} = \U_i^*\mathcal{W}_i\bar{\Gamma}^*$
and $\wsetD{T,i}{<}{1} = \U_i^*\mathcal{W}_i\bar{\Gamma}^* \cup  \U_i^*$.
The ``$\subseteq$'' direction of Proposition~\ref{prop:cor_alg} is 
implied by the following lemma:
\begin{lemma}
\label{lem-subset-prelim}
For every $0 \leq i \leq K$, there are SMD strategies 
$\pi[\W_i],\pi[\U_i] \in \Pi_0$ constructible by a polynomial-time
algorithm with $\NP \cap \coNP$ oracle such that
\begin{itemize}
\item[(1)] For every $\bar{X} \in \mathcal{W}_i$ and every 
  $\sigma_0 \in \Sigma_0$ we have that 
  $\calP^{\sigma_0,\pi[\W_i]}_{\bar{X}}(\R(T_\varepsilon),G_{\Theta_0}) < 1$
  or 
  $\calP^{\sigma_0,\pi[\W_i]}_{\bar{X}}(\R(\wsetD{T_\varepsilon,i{-}1}{<}{1},G_{\Theta_0})) > 0$
\item[(2)] For every $\bar{Y} \in \U_i$ and every 
  $\sigma_0 \in \Sigma_0$ we have that 
  $\calP^{\sigma_0,\pi[\U_i]}_{\bar{Y}}(\R(T),G_{\Theta_0}) < 1$
  or   
  $\calP^{\sigma_0,\pi[\U_i]}_{\bar{Y}}(\R(\wsetD{T,i{-}1}{<}{1},G_{\Theta_0})) > 0$
\item[(3)] If $i > 0$, then $\pi[\W_i](\bar{X}) = \pi[\W_{i{-}1}](\bar{X})$ 
  for every
  $\bar{X} \in \W_{i-1}$ and
  $\pi[\U_i](\bar{Y}) = \pi[\U_{i{-}1}](\bar{Y})$ for every
  $\bar{Y} \in \U_{i-1}$.
\end{itemize}
\end{lemma}
\begin{proof}
The strategies $\pi[\W_i],\pi[\U_i]$ are constructed inductively on~$i$. 
In the base case, $\pi[\W_0],\pi[\U_0]$ are chosen arbitrarily. Now assume
that $\pi[\W_i],\pi[\U_i] \in \Pi_0$ have already been constructed.
Due to Lemma~\ref{lem-witness-strategy}, there is a SMD 
strategy $\pi_i \in \Pi_i$ constructible by a deterministic
polynomial-time algorithm with $\NP \cap \coNP$ oracle 
such that for every $\bar{Z} \in \M_i$ and every $\sigma_i \in \Sigma_i$ 
we have that $\calP^{\sigma_i,\pi_i}_{\bar{Z}}(\R(T_\varepsilon,G_{\Theta_i})) < 1$.
(Strictly speaking, Lemma~\ref{lem-witness-strategy} guarantees
the existence of a SMD strategy $\pi_i \in \Pi_i$ such that the
above condition is satisfied just for all $\bar{Z} \in W_i$. However,
the strategy $\pi_i$ of Lemma~\ref{lem-witness-strategy} can be easily 
modified so that it works for all $\bar{Z} \in \M_i =
\bigcup_{j=0}^{|\bar{\Gamma}|} \FuncSty{Att}^j_{\Theta_i,W_i}(\emptyset)$;
whenever a new symbol $\bar{A} \in
\bar{\Gamma}_\Diamond$ appears in $\FuncSty{Att}^{j+1}_{\Theta_i,W_i}(\emptyset)$, we 
fix one of the rules $\bar{A} \bbtran[i]{} \bar{B}$ which witness  
the membership of $\bar{A}$ to $\FuncSty{Att}^{j+1}_{\Theta_i,W_i}(\emptyset)$.) 
The strategies $\pi[\W_{i{+}1}]$ and $\pi[\U_{i{+}1}]$ are defined as follows: 
 \begin{itemize}
\item for every $\bar{X} \in \W_i$, we put 
  $\pi[\W_{i{+}1}](\bar{X}) = \pi[\W_i](\bar{X})$;
\item for every $\bar{X} \in \W_{i{+}1} \smallsetminus
  \W_i = \M_i$, we put $\pi[\W_{i{+}1}](\bar{X}) =
  \pi_i(\bar{X})$;
\item for every $\bar{Y} \in \U_{i}$, we put 
  $\pi[\U_{i{+}1}](\bar{Y}) = \pi[\U_i](\bar{Y})$;
\item for every $\bar{Y} \in \U_{i{+}1} \smallsetminus \U_i$, the 
  distribution $\pi[\U_{i{+}1}](\bar{Y})$ selects the
  (unique) rule $\bar{Y} \bbtran[0]{} \bar{Q}$ such that $\pp{Y}
  \bbtran[i]{} \pp{Q}$ is the rule selected by $\pi_i(\pp{Y})$.
\end{itemize}
Observe that for every $0\leq i \leq K$, the strategies 
$\pi[\W_i],\pi[\U_i]$ are constructible by a deterministic polynomial-time
algorithm with $\NP \cap \coNP$ oracle.

Now we show that Conditions~(1)--(3) are satisfied for every
$0 \leq i \leq K$. We proceed by induction on $i$.
The base case ($i=0$) is immediate, because 
$\W_0 = \U_0 = \emptyset$. Now let us assume that
$\pi[\W_i],\pi[\U_i]$ satisfy Conditions~(1)--(3). The
strategies $\pi[\W_{i{+}1}],\pi[\U_{i{+}1}]$ obviously satisfy
Condition~(3). By induction hypothesis, 
Condition~(1) and Condition~(2) are satisfied for all elements of
$\mathcal{W}_i$ and $\U_i$, respectively. We verify that
Condition~(1) and Condition~(2) are satisfied also for the remaining
symbols of $\mathcal{W}_{i{+}1} \smallsetminus \mathcal{W}_i$ and 
$\U_{i{+}1} \smallsetminus \U_i$, respectively.

\emph{Condition~(1)}. Let us fix some 
$\bar{X} \in \mathcal{W}_{i{+}1} \smallsetminus \mathcal{W}_i$ and
$\sigma_0 \in \Sigma_0$. If 
$\calP^{\sigma_0,\pi[\W_{i{+}1}]}_{\bar{X}}(\R(\wsetD{T_\varepsilon,i}{<}{1},
G_{\Theta_0})) > 0$,
we are done. Now assume that 
$\calP^{\sigma_0,\pi[\W_{i{+}1}]}_{\bar{X}}(\R(\wsetD{T_\varepsilon,i}{<}{1},
G_{\Theta_0})) = 0$.
We show that the strategy $\sigma_0$ can be ``mimicked'' by a strategy 
$\sigma_{i} \in \Sigma_{i}$ so that
\begin{equation}
  \calP^{\sigma_0,\pi[\W_{i{+}1}]}_{\bar{X}}(\R(T_\varepsilon),\Theta_0) \ = \
  \calP^{\sigma_{i},\pi_{i}}_{\bar{X}}(\R(T_\varepsilon),\Theta_{i}) \ > \ 0 
\end{equation}
We construct the strategy $\sigma_{i}$ so that the reachable parts of 
the plays $G_{\Theta_0}(\sigma_0,\pi[\W_{i{+}1}])$ and
$G_{\Theta_{i}}(\sigma_{i},\pi_{i})$ initiated in $\bar{X}$ become
isomorphic. Let \mbox{$f : \bar{\Gamma}^* \rightarrow \bar{\Gamma}_{i}^*$} 
be a function defined inductively as follows:
\begin{itemize}
\item $f(\varepsilon) = \varepsilon$;
\item if $\bar{Y}\beta \in \wsetD{T_\varepsilon,i}{<}{1}$, then
  $f(\bar{Y}\beta) = f(\beta)$;
\item if $\bar{Y}\beta,\beta \not\in \wsetD{T_\varepsilon,i}{<}{1}$, 
  then $\bar{Y} \in \bar{\Gamma}_{i}$ (because $\bar{Y}
  \not\in \W_i$) and we put $f(\bar{Y}\beta) =
  \bar{Y}f(\beta)$;
\item if $\bar{Y}\beta \not\in \wsetD{T_\varepsilon,i}{<}{1}$
  and $\beta \in \wsetD{T_\varepsilon,i}{<}{1}$, then $\pp{Y}
  \in \bar{\Gamma}_i$ and we put $f(\bar{Y}\beta) = \pp{Y}f(\beta)$ (observe
  that if $\pp{Y} \not\in \bar{\Gamma}_i$, then $\pp{Y} \in \mathcal{W}_i$
  and $\bar{Y} \in \U_i$, which contradicts the assumption that
  $\bar{Y}\beta \not\in \wsetD{T_\varepsilon,i}{<}{1}$).
\end{itemize}
For every reachable state $\alpha_0,\ldots,\alpha_j$ of
$G_{\Theta_0}(\sigma_0,\pi[\W_{i{+}1}])$ we put $\F(\alpha_0,\ldots,\alpha_j) =
f(\alpha_0),\ldots,f(\alpha_j)$ where $f$ is the function defined
above. Our aim is to setup the strategy $\sigma_i$ so that $\F$
becomes an isomorphism. This means to ensure that for every reachable
state $\alpha_0,\ldots,\alpha_j$ of $G_{\Theta_0}(\sigma_0,\pi[\W_{i{+}1}])$ we
have that $f(\alpha_0),\ldots,f(\alpha_j)$ is a reachable state of
$G_{\Theta_i}(\sigma_i,\pi_i)$, and
$\alpha_0,\ldots,\alpha_{j-1} \tran{x} \alpha_0,\ldots,\alpha_j$
implies $f(\alpha_0),\ldots,f(\alpha_{j-1}) \tran{x}
f(\alpha_0),\ldots,f(\alpha_j)$. 
We proceed by induction on~$j$ and
define the strategy $\sigma_i$ on the fly so that the above
condition is satisfied.  The base case (when $j=0$) is immediate,
because $\F(\bar{X}) = f(\bar{X}) = \bar{X}$ and the root $\bar{X}$ has
no incoming transitions. Now assume that $\alpha_0,\ldots,\alpha_{j}$
is a reachable state of $G_{\Theta_0}(\sigma_0,\pi[\W_{i{+}1}])$ such that
$\alpha_0,\ldots,\alpha_{j-1} \tran{x} \alpha_0,\ldots,\alpha_j$.
Then $\alpha_{j-1} \gtran[0]{} \alpha_j$ is an edge in $G_{\Theta_0}$,
which is assigned the probability~$x$ either by $\Prob$, $\pi_0$, or
$\sigma_0$, depending on whether the first symbol of $\alpha_{j-1}$
belongs to $\bar{\Gamma}_\bigcirc$, $\bar{\Gamma}_\Diamond$, or
$\bar{\Gamma}_\Box$, respectively. By induction hypothesis,
$f(\alpha_0),\ldots,f(\alpha_{j-1})$ is a reachable state of
$G_{\Theta_i}(\sigma_i,\pi_i)$ and hence it suffices to show
that $f(\alpha_{j-1}) \gtran[i]{} f(\alpha_j)$ is an edge in
$G_{\Theta_i}$ which is assigned the same probability~$x$ by $\Prob$,
$\pi_i$, or the newly constructed $\sigma_i$, respectively.  Let
$\alpha_{j-1} = \bar{A}\beta$.  Note that since $\bar{A}\beta \not\in
\wsetD{T_\varepsilon,i}{<}{1}$, we have that $f(\alpha_{j-1}) =
f(\bar{A}\beta) = \hat{A}f(\beta)$, where $\hat{A} = \pp{A}$ or
$\hat{A} = A$ depending on whether $\beta \in \wsetD{T_\varepsilon,i}{<}{1}$ 
or not, respectively. If $\bar{A} \in
\bar{\Gamma}_\bigcirc$, then $\alpha_j = \bar{B}\beta$ for some
$\bar{B}$ such that $\bar{A} \bbtran[0]{x} \bar{B}$. But then also
$\hat{A} \bbtran[i]{x} \hat{B}$, where $\hat{B}$ is either $B$ or
$\pp{B}$ depending on whether $\hat{A} = A$ or $\hat{A} = \pp{A}$,
respectively. Hence, $f(\alpha_{j-1}) = \hat{A}f(\beta) \gtran[i]{x}
\hat{B}f(\beta) = f(\bar{B}\beta) = f(\alpha_j)$ as needed.  If
$\bar{A} \in \bar{\Gamma}_\Diamond$, we argue in a similar way, using
the definitions of $\pi[\W_{i{+}1}]$ and $\pi_i$. The most complicated case is
when $\bar{A} \in \bar{\Gamma}_\Box$. It suffices to show that
$f(\alpha_{j-1}) = \hat{A}f(\beta) \gtran[i]{} f(\alpha_j)$.  The
distribution $\sigma_i(f(\alpha_0),\ldots,f(\alpha_{j-1}))$ can
then safely select the edge $f(\alpha_{j-1}) \gtran[i]{} f(\alpha_j)$
with probability~$x$. According to Definition~\ref{ass:rules}, we can
distinguish the following three possibilities:
\begin{itemize}
\item $\bar{A} \in \bar{\Gamma}[1]$. Then $\alpha_j = \bar{B}\beta$
  for some $\bar{B}$ such that $\bar{A} \bbtran[0]{} \bar{B}$.  If
  $\hat{A} = A$, then $\bar{A} = A$, $\bar{B} = B$, and $\alpha
  \not\in \wsetD{T_\varepsilon,i}{<}{1}$ by the definition
  of~$f$. Further, $B \in \Gamma_i$ because otherwise $B \in
  \mathcal{W}_i$ and $B\beta \in \wsetD{T_\varepsilon,i}{<}{1}$,
  which contradicts the assumption
  $\calP^{\sigma_0,\pi[\W_{i{+}1}]}_{\bar{X}}(\R(\wsetD{T_\varepsilon,i}{<}{1}), \Theta_0) = 0$. 
  Hence, $f(\alpha_{j-1}) = Af(\beta)
  \gtran[i]{} Bf(\beta) = f(\alpha_j)$ as needed.

  If $\hat{A} = \pp{A}$, then either $\bar{A} = A$ or $\bar{A} =
  \pp{A}$, and we consider these two cases separately.  If $\bar{A} =
  A$, then $\bar{B} = B$ and $\alpha \in \wsetD{T_\varepsilon,i}{<}{1}$ 
  by the definition of~$f$. Further, $\pp{B} \in
  \Gamma_i$ because otherwise $B \in \U_i$ and thus 
  $B\beta \in \wsetD{T_\varepsilon,i}{<}{1}$, which contradicts the assumption
  $\calP^{\sigma_0,\pi[\W_{i{+}1}]}_{\bar{X}}(\R(\wsetD{T_\varepsilon,i}{<}{1}), 
  \Theta_0) = 0$. 
   Hence, $f(\alpha_{j-1}) =
  \pp{A}f(\beta) \gtran[i]{} \pp{B}f(\beta) = f(B\beta) =
  f(\alpha_j)$.
  If $\bar{A} = \pp{A}$, then $\bar{B} = \pp{B}$ and $\pp{B} \in
  \Gamma_i$, because otherwise $\pp{B} \in \mathcal{W}_i$ and
  $\pp{B}\beta \in \wsetD{T_\varepsilon,i}{<}{1}$, which
  contradicts the assumption
  $\calP^{\sigma_0,\pi[\W_{i{+}1}]}_{\bar{X}}(\R(\wsetD{T_\varepsilon,i}{<}{1}),
  \Theta_0) = 0$.
   Hence, $f(\alpha_{j-1}) =
  \pp{A}f(\beta) \gtran[i]{} \pp{B}f(\beta) = f(\alpha_j)$.
\item $\bar{A} \in \bar{\Gamma}[2]$. Then $\alpha_j = \beta$. If
  $\hat{A} = A$, then $\beta \not\in \wsetD{T_\varepsilon,i}{<}{1}$ 
  and $f(\alpha_{j-1}) = Af(\beta) \gtran[i]{}
  f(\beta) = f(\alpha_j)$.  If $\hat{A} = \pp{A}$, then $\beta \in
  \wsetD{T_\varepsilon,i}{<}{1}$ by the definition of~$f$ which
  contradicts the assumption
  $\calP^{\sigma_0,\pi[\W_{i{+}1}]}_{\bar{X}}(\R(\wsetD{T_\varepsilon,i}{<}{1}), \Theta_0) = 0$.
\item $\bar{A} \in \bar{\Gamma}[3]$. If $\hat{A} = A$, then $\bar{A} =
  A$ and $\alpha_j = BC\beta$ where $A \bbtran[0]{} BC$ is the only
  available rule with $A$ on the left-hand side.  Further, we either
  have $A \bbtran[i]{} BC$ or $A \bbtran[i]{} \pp{B}$.  In the first
  case we obtain $B,C \not\in \mathcal{W}_i$, which means that
  $f(BC\beta) = BCf(\beta)$ and hence $f(\alpha_{j-1}) = Af(\beta)
  \gtran[i]{} BCf(\beta) = f(\alpha_j)$.  In the latter case, $C \in
  \mathcal{W}_i$ and $\pp{B} \not\in \mathcal{W}_i$, which means that
  $B \not\in \U_i$.  Hence, $f(BC\beta) = \pp{B}f(\beta)$ and
  $f(\alpha_{j-1}) = Af(\beta) \gtran[i]{} \pp{B}f(\beta) =
  f(\alpha_j)$.

  If $\hat{A} = \pp{A}$, then either $\bar{A} = A$ or $\bar{A} =
  \pp{A}$. If $\bar{A} = A$, then $\beta \in \wsetD{T_\varepsilon,i}{<}{1}$ 
  and $\alpha_j = BC\beta$ where $A \bbtran[0]{} BC$
  is the only available rule with $A$ on the left-hand side.  Further,
  we either have $\pp{A} \bbtran[i]{} \pp{B}$ or $\pp{A} \bbtran[i]{}
  B\pp{C}$. In the first case, we have that $\pp{C} \in
  \mathcal{W}_i$, hence $C \in \U_i$ and $C\beta \in 
  \wsetD{T_\varepsilon,i}{<}{1}$. This means $f(BC\beta) =
  \pp{B}f(\beta)$ and hence $f(\alpha_{j-1}) = \pp{A}f(\beta)
  \gtran[i]{} \pp{B}f(\beta) = f(\alpha_j)$.  In the latter case,
  $B,\pp{C} \not\in \mathcal{W}_i$, which means that $\pp{C} \not\in
  \mathcal{W}_i$ and hence $C \not\in \U_i$. Now realize that for
  every $P \in \Gamma$ we have that if $P \in \mathcal{W}_i$, then
  $\pp{P} \in \U_i$.  This follows directly from the ``main''
  induction hypothesis (which states that Conditions~(1) and~(2) hold
  for the symbols of $\mathcal{W}_i$ and $\U_i$, respectively) and
  Lemma~\ref{lem-relationship}. From this and $C \not\in \U_i$ we can
  conclude that $C \not\in \mathcal{W}_i$. This implies $C\beta
  \not\in \wsetD{T_\varepsilon,i}{<}{1}$, which means that
  $f(BC\beta) = B\pp{C}f(\beta)$ and hence $f(\alpha_{j-1}) =
  \pp{A}f(\beta) \gtran[i]{} B\pp{C}f(\beta) = f(\alpha_j)$.
\end{itemize}

\emph{Condition~(2)}. We proceed similarly as in the case of
Condition~(1).  Let $\bar{Y} \in \U_{i{+}1} \smallsetminus \U_i$ and
$\sigma_0 \in \Sigma_0$. If 
$\calP^{\sigma_0,\pi[\U_{i{+}1}]}_{\bar{X}}(\R(\wsetD{T,i}{<}{1},
G_{\Theta_0})) > 0$, we are done. Now  assume that 
$\calP^{\sigma_0,\pi[\U_{i{+}1}]}_{\bar{X}}(\R(\wsetD{T,i}{<}{1},
G_{\Theta_0})) = 0$. We show that the strategy $\sigma_0$ can be 
``mimicked'' by a strategy 
$\hat{\sigma}_{i} \in \Sigma_{i}$ so that
\begin{equation}
  \calP^{\sigma_0,\pi[\U_{i{+}1}]}_{\bar{Y}}(\R(T),\Theta_0) \ = \
  \calP^{\hat{\sigma}_{i},\pi_{i}}_{\pp{Y}}(\R(T),\Theta_{i}) \ > \ 0 
\end{equation}
We construct the strategy $\hat{\sigma}_{i}$ so that the reachable parts of 
the plays $G_{\Theta_0}(\sigma_0,\pi[\U_{i{+}1}])$ and
$G_{\Theta_{i}}(\hat{\sigma}_{i},\pi_{i})$ initiated in $\bar{Y}$ and
$\pp{Y}$ become
isomorphic. Let \mbox{$\hat{f} : \bar{\Gamma}^* \rightarrow \bar{\Gamma}_{i}^*$} 
be a function defined in the same way as $f$ except that 
$\wsetD{T,i}{<}{1}$ is used instead of $\wsetD{T_\varepsilon,i}{<}{1}$.
For every reachable state $\alpha_0,\ldots,\alpha_j$ of
$G_{\Theta_0}(\sigma_0,\pi[\U_{i{+}1}])$ we put 
$\hat{\F}(\alpha_0,\ldots,\alpha_j) = 
\hat{f}(\alpha_0),\ldots,\hat{f}(\alpha_j)$
and define the strategy $\hat{\sigma}_i$ so that the function 
$\hat{\F}$ becomes an isomorphism. The rest of the proof is almost the same
as for Condition~(1).
\end{proof}

Note that an immediate consequence of Lemma~\ref{lem-subset-prelim} is 
the following:
\begin{lemma}
\label{lem-subset}
For every $0 \leq i \leq K$ we have that 
$\mathcal{W}_i \subseteq \wsetD{T_\varepsilon}{<}{1}$ and
$\U_i \subseteq \wsetD{T}{<}{1}$ 
\end{lemma}
Lemma~\ref{lem-subset} is proven by a trivial induction on~$i$,
using Lemma~\ref{lem-subset-prelim} in the induction step.
Thus, the ``$\subseteq$'' direction of Proposition~\ref{prop:cor_alg} 
is established. The opposite direction is shown in the next lemma.

\begin{lemma}
\label{lem-superset}
  We have that $\mathcal{W}_K \cap \A_0 = \emptyset$ and 
  $\U_K \cap \C_0 = \emptyset$. 
\end{lemma}
\begin{proof}
Since $W_K = \emptyset$, due
to Proposition~\ref{prop:no-wit-reg-win} there is a regular MR strategy
$\sigma_K \in \Sigma_K$ which is \mbox{$(T_\varepsilon,{=}1)$-winning}
in every $\alpha \in \bar{\Gamma}_K^*$. Moreover, the strategy 
$\sigma_K$ is computable in time which is polynomial in the size
of $\Theta_K$ (assuming that $\Theta_K$ has already been computed).
Let $\B_K = \bar{\Gamma} \smallsetminus \mathcal{W}_K = \bar{\Gamma}_K$
and $\scrD_K = \bar{\Gamma} \smallsetminus \U_K$. We show that the strategy
$\sigma_K$ can be efficiently transformed into regular MR strategies
$\sigma_0,\hat{\sigma}_0 \in \Sigma_0$ such that 
$\sigma_0$ is \mbox{$(T_\varepsilon,{=}1)$-winning}
in every configuration of $\B_K^*$, and $\hat{\sigma}_0$ is 
\mbox{$(T,{=}1)$-winning}
in every configuration of $\B_K^*\scrD_K\bar{\Gamma}$. In particular,
this means that $\B_K \subseteq \wsetB{T_\varepsilon}{=}{1}$ and
$\scrD_K \subseteq \wsetB{T}{=}{1}$, hence 
$\mathcal{W}_K \cap \A_0 = \emptyset$ and $\U_K \cap \C_0 = \emptyset$ 
as needed. 

First we show how to construct the strategy $\sigma_0$. 
We start by defining (partial) functions 
$g,h : \bar{\Gamma}^* \rightarrow \bar{\Gamma}_K^*$ inductively as follows:
\begin{itemize}
\item $g(\varepsilon) = h(\varepsilon) = \varepsilon$
\item $g(\bar{Y}\beta) = \begin{cases}
                           \bar{Y}g(\beta) & \mbox{if } \bar{Y} \in \B_K
                           \mbox{ and } 
                           \beta \in \B_K\bar{\Gamma}^* \cup \{\varepsilon\};\\
                           \pp{Y}h(\beta) & \mbox{if }  \bar{Y},\pp{Y} 
                           \in \B_K \mbox{ and } 
                           \beta \in (\bar{\Gamma} \smallsetminus \B_K)
                                  \bar{\Gamma}^*;\\
                           \perp & \mbox{otherwise.}
                         \end{cases}$
\item $h(\bar{Y}\beta) = \begin{cases}
                           g(\bar{Y}\beta) & \mbox{if } \bar{Y} \in \B_K;\\
                           h(\beta) & \mbox{otherwise.}
                         \end{cases}$
\end{itemize}
A configuration $\alpha \in \bar{\Gamma}^*$ is called \emph{$g$-eligible}
if $g(\alpha) \neq {\perp}$. The strategy $\sigma_0$ is constructed so that 
for every $g$-eligible $\bar{A}\alpha \in \bar{\Gamma}_\Box\bar{\Gamma}^*$, 
the following conditions are satisfied:
\begin{itemize}
 \item If $\bar{A} \in \bar{\Gamma}[1]$ and $\sigma_K(g(\bar{A}\alpha))$
  selects a rule $\hat{A} \bbtran[K]{} \hat{B}$ with
  probability~$x$, then $\sigma_0(\bar{A}\alpha)$ selects the rule 
  $\bar{A} \bbtran[0]{} \bar{B}$ with probability~$x$.
 \item If $\bar{A} \in \bar{\Gamma}[2] \cup \bar{\Gamma}[3]$, then 
  $\sigma_0(\bar{A}\alpha)$ selects the only available rule
  with probability~$1$. 
\end{itemize}
Note that the definition of $\sigma_0$ is effective in the sense that
if the finite-state automaton $\A_{\sigma_K}$ associated with the 
regular MR strategy $\sigma_K$ (see Definition~\ref{def-regular-strategy}) 
has already been computed, then the finite-state automaton 
$\A_{\sigma_0}$ associated with $\sigma_0$ simply ``simulates'' the 
execution of $\A_{\sigma_K}$ on the reverse of $g(\alpha)$ for every
$g$-eligible $\alpha \in \bar{\Gamma}^*$. Hence, the automaton 
$\A_{\sigma_0}$ is constructible in polynomial time assuming that
the BPA game $\Theta_K$ has already been computed 
(cf.~Proposition~\ref{prop:no-wit-reg-win}).

We show that for every $g$-eligible initial configuration 
$\gamma \in \bar{\Gamma}^*$ and every $\pi_0 \in \Pi_0$ we have that 
$\calP^{\sigma_0,\pi_0}_{\gamma}(\R(T_\varepsilon),\Theta_0) = 1$.
Assume the converse, i.e., there is a strategy $\pi_0 \in \Pi_0$ and 
a $g$-eligible configuration $\gamma$ such that
$\calP^{\sigma_0,\pi_0}_{\gamma}(\R(T_\varepsilon),\Theta_0) < 1$. 
We show that then there is a strategy $\pi_K \in \Pi_k$ such that
\begin{equation}
\calP^{\sigma_0,\pi_0}_{\gamma}(\R(T_\varepsilon),\Theta_0) \ = \
\calP^{\sigma_K,\pi_K}_{g(\gamma)}(\R(T_\varepsilon),\Theta_K) \ = \ 1
\end{equation}
which is a contradiction. For every finite sequence of $g$-eligible 
configurations $\alpha_0,\ldots,\alpha_n \in \bar{\Gamma}^*$
such that $\alpha_n = \bar{A}\beta \in \bar{\Gamma}_{\Diamond}\bar{\Gamma}^*$,
the strategy $\pi_K$ selects a rule $\hat{A} \bbtran[K]{} \hat{B}$
in $g(\alpha_0),\ldots,g(\alpha_n)$ with probability~$x$ iff 
the strategy $\pi_0$ selects a rule $\bar{A} \bbtran[K]{} \bar{B}$
in $\alpha_0,\ldots,\alpha_n$ with probability~$x$.
We show that every reachable state $\alpha_0,\ldots,\alpha_j$ of the play 
$G_{\Theta_0}(\sigma_0,\pi_0)$ initiated in $\gamma$ is a sequence
of $g$-eligible configurations and the function $\G$ over the reachable 
states of $G_{\Theta_0}(\sigma_0,\pi_0)$ defined by
$\G(\alpha_0,\ldots,\alpha_j) = g(\alpha_0),\ldots,g(\alpha_j)$ is an
isomorphism between the reachable parts of the plays 
$G_{\Theta_0}(\sigma_0,\pi_0)$  and $G_{\Theta_K}(\sigma_K,\pi_K)$
initiated in $\gamma$ and $g(\gamma)$, respectively. 
We proceed by induction on~$j$. The base case is immediate. Now assume that 
$\alpha_0,\ldots,\alpha_{j}$ is a reachable state of 
$G_{\Theta_0}(\sigma_0,\pi_0)$ such that
$\alpha_0,\ldots,\alpha_{j-1} \tran{x} \alpha_0,\ldots,\alpha_j$.
Then $\alpha_{j-1} \gtran[0]{} \alpha_j$ is an edge in $G_{\Theta_0}$,
which is assigned the probability~$x$ either by $\Prob$, $\pi_0$,
or $\sigma_0$, depending on whether the first symbol of
$\alpha_{j-1}$ belongs to $\bar{\Gamma}_\bigcirc$, $\bar{\Gamma}_\Diamond$,
or $\bar{\Gamma}_\Box$, respectively. By induction hypothesis,
$\alpha_0,\ldots,\alpha_{j-1}$ is a sequence of $g$-eligible states
and hence it suffices to show that
$g(\alpha_{j-1}) \gtran[K]{} g(\alpha_j)$ is an edge in $G_{\Theta_K}$
which is assigned the same probability~$x$ by $\Prob$, $\pi_K$,
or $\sigma_K$, respectively. Let $\alpha_{j-1} = \bar{A}\beta$. 
We distinguish three possibilities:
\begin{itemize}
\item $\bar{A} \in \bar{\Gamma}_\bigcirc \cup \bar{\Gamma}_\Diamond
  \cup \bar{\Gamma}[1]$. Then $\alpha_{j-1} = \bar{A}\beta \gtran[0]{}
  \bar{B}\beta = \alpha_j$ where $\bar{A} \bbtran[0]{} \bar{B}$. If
  $g(\bar{A}\beta) = \bar{A}g(\beta)$, then $\beta \in
  \B_K\bar{\Gamma}^* \cup \{\varepsilon\}$ and since $\bar{A}
  \bbtran[K]{} \bar{B}$, we have that $g(\bar{A}\beta) =
  \bar{A}g(\beta) \gtran[K]{} \bar{B}g(\beta) = g(\bar{B}\beta)$. If
  $g(\bar{A}\beta) = \pp{A}h(\beta)$, then $\pp{A} \in \B_K$ and
  $\beta \in (\bar{\Gamma} \smallsetminus \B_K)\bar{\Gamma}^*$, hence
  $\pp{A} \bbtran[K]{} \pp{B}$ and $g(\bar{A}\beta) = \pp{A}h(\beta)
  \gtran[K]{} \pp{B}h(\beta) = g(\bar{B}\beta)$. It follows immediately
  from the definition of $\sigma_0$ and $\pi_K$ that the edges
  $\alpha_{j-1} \gtran[0]{} \alpha_j$ and 
  $g(\alpha_{j-1}) \gtran[K]{} g(\alpha_j)$ are assigned the same 
  probability.
\item $\bar{A} \in \bar{\Gamma}[2]$. Then $\bar{A} = A$ and 
  $\alpha_{j-1} = A\beta \gtran[0]{} \beta = \alpha_j$. Further,
  observe that $\pp{A} \not\in \bar{\Gamma}_K$, because otherwise
  $\pp{A} \bbtran[K]{} \pp{A}$ is the only rule with $\pp{A}$ on 
  the left-hand side, hence $\pp{A} \in W_K$ and we have a contradiction.
  So, $g(A\beta) = Ag(\beta) \gtran[K]{} g(\beta)$. Obviously,
  $\alpha_{j-1} \gtran[0]{} \alpha_j$ and 
  $g(\alpha_{j-1}) \gtran[K]{} g(\alpha_j)$ are assigned 
  probability~$1$.
\item $\bar{A} \in \bar{\Gamma}[3]$. Then 
  $\alpha_{j-1} = \bar{A}\beta \gtran[0]{}
  B\bar{C}\beta = \alpha_j$ where $\bar{A} \bbtran[0]{} B\bar{C}$.
  If $g(\bar{A}\beta) = \bar{A}g(\beta)$, then $\beta \in
  \B_K\bar{\Gamma}^* \cup \{\varepsilon\}$ and there are two
  possibilities:
  \begin{itemize}
  \item $\bar{A} \bbtran[K]{} B\bar{C}$. By the definition of $g$,
    we have that $g(B\bar{C}\beta) = B\bar{C}g(\beta)$, hence
    $g(\bar{A}\beta) = \bar{A}g(\beta) \gtran[K]{} B\bar{C}g(\beta)=
    g(B\bar{C}\beta)$ as needed.
  \item $\bar{A} \bbtran[K]{} \pp{B}$. Then $\bar{A},\pp{A} \in \B_K$ and
    $\bar{C} \not\in \B_K$,
    hence $g(B\bar{C}\beta) = \pp{B}h(\bar{C}\beta)= \pp{B}g(\beta)$
    by the definition of $g$. Thus, $g(\bar{A}\beta) = \bar{A}g(\beta) 
    \gtran[K]{} \pp{B}g(\beta)= g(B\bar{C}\beta)$.
  \end{itemize}
  If $g(\bar{A}\beta) = \pp{A}h(\beta)$, then $\pp{A} \in \B_K$ and
  $\beta \in (\bar{\Gamma} \smallsetminus \B_K)\bar{\Gamma}^*$. Again,
  there are two possibilities:
  \begin{itemize}
  \item $\pp{A} \bbtran[K]{} B\pp{C}$. By the definition of $g$,
    we have that $g(B\pp{C}\beta) = B\pp{C}h(\beta)$, hence
    $g(\bar{A}\beta) = \pp{A}h(\beta) \gtran[K]{} B\pp{C}h(\beta)=
    g(B\bar{C}\beta)$.
  \item $\pp{A} \bbtran[K]{} \pp{B}$. Then $\bar{A},\pp{A} \in \B_K$ and
    $\pp{C} \not\in \B_K$,
    hence $g(B\pp{C}\beta) = \pp{B}h(\pp{C}\beta)= \pp{B}h(\beta)$
    by the definition of $g$. Thus, $g(\bar{A}\beta) = \pp{A}h(\beta) 
    \gtran[K]{} \pp{B}h(\beta)= g(B\pp{C}\beta)$.
  \end{itemize}
  In all of the above discussed subcases, we have that the edges
  $\alpha_{j-1} \gtran[0]{} \alpha_j$ and 
  $g(\alpha_{j-1}) \gtran[K]{} g(\alpha_j)$ are assigned 
  probability~$1$.
\end{itemize}
Since every configuration of $\B_K^*$ is $g$-eligible, the strategy
$\sigma_0$ is \mbox{$(T_\varepsilon,{=}1)$-winning}
in every configuration of $\B_K^*$.

The definition of $\hat{\sigma}_0$ and the proof that $\hat{\sigma}_0$
is \mbox{$(T,{=}1)$-winning} in every configuration of 
$\B_K^*\D_K\bar{\Gamma}^*$ are very similar as in the case of $\sigma_0$. 
The main (and only) difference is the definition of the function~$g$. 
Instead of $g$ and $h$, we use partial functions 
$\hat{g},\hat{h} : \bar{\Gamma}^* \rightarrow \bar{\Gamma}_K^*$ 
defined  as follows:
\begin{itemize}
\item $\hat{g}(\varepsilon) = {\perp}$, $h(\varepsilon) = \varepsilon$
\item $\hat{g}(\bar{Y}\beta) = \begin{cases}
                           \bar{Y}\hat{g}(\beta) & \mbox{if } \bar{Y} \in \B_K
                           \mbox{ and } 
                           \beta \in \B_K^*\scrD_K \bar{\Gamma}^*;\\
                           \pp{Y}\hat{h}(\beta) & \mbox{if }  \bar{Y},\pp{Y} 
                           \in \B_K \mbox{ and } 
                           \beta \not\in \B_K^*\scrD_K \bar{\Gamma}^*;\\
                           \perp & \mbox{otherwise.}
                         \end{cases}$
\item $\hat{h}(\bar{Y}\beta) = \begin{cases}
                           \hat{g}(\bar{Y}\beta) & 
                               \mbox{if } \bar{Y} \in \B_K;\\
                           \hat{h}(\beta) & \mbox{otherwise.}
                         \end{cases}$
\end{itemize}
The strategy $\hat{\sigma}_0$ and the function $\hat{\G}$ are defined
in the same way as $\hat{\sigma}_0$ and $\G$, using $\hat{g}$ instead
of~$g$. The strategy $\pi_K$ is defined in the same way as above.
Observe that 
$\calP^{\sigma_K,\pi_K}_{\hat{g}(\gamma)}(\R(T_\varepsilon),\Theta_K) = 1$
and since $g(\gamma)$ contains at least one symbol of $\hat{\Gamma}$,
we have that 
$\calP^{\sigma_K,\pi_K}_{\hat{g}(\gamma)}(\R(\varepsilon),\Theta_K) = 0$ which
means that $\calP^{\sigma_K,\pi_K}_{\hat{g}(\gamma)}(\R(T),\Theta_K) = 1$. 
The case analysis which reveals that $\hat{\G}$ is an isomorphism
between the reachable parts of the
plays $G_{\Theta_0}(\hat{\sigma}_0,\pi_0)$  and $G_{\Theta_K}(\sigma_K,\pi_K)$
initiated in $\gamma$ and $\hat{g}(\gamma)$, respectively, is almost
the same as above. 
\end{proof}

Lemma~\ref{lem-subset} and Lemma~\ref{lem-superset} together imply
Proposition~\ref{prop:cor_alg}. It remains to prove 
Theorem~\ref{thm:box_strat_1}. The strategy $\hat{\sigma}_0$ 
constructed in the proof of Lemma~\ref{lem-superset} is
$(T,{=}1)$-winning in every configuration of 
$\B_K^*\scrD_K\bar{\Gamma}^*$. Since $\B_K = \B$ and $\scrD_K = \scrD$ by
Proposition~\ref{prop:cor_alg}, the strategy $\hat{\sigma}_0$
is $(T,{=}1)$-winning in every configuration of 
$\wsetB{T}{=}{1}$. As it was noted in the proof of Lemma~\ref{lem-superset},
the strategy $\hat{\sigma}_0$ is constructible in polynomial time
assuming that the BPA game $\Theta_K$ has already been computed.
Since $\Theta_K$ is computable by a deterministic polynomial-time algorithm with
$\NP \cap \coNP$ oracle, the first part of Theorem~\ref{thm:box_strat_1}
is proven. It remains to show that there is a regular strategy
$\pi \in \Pi$ constructible by a deterministic polynomial-time algorithm
with $\NP \cap \coNP$ oracle such that $\pi$ is $(T,{<}1)$-winning 
in every configuration of 
$\wsetD{T}{<}{1} = \C^* \A \bar{\Gamma}^* \cup \C^*$. 
For all $\bar{X} \in \A$ and $\bar{Y} \in \C$, let
$I_{\A}(\bar{X})$ and $I_{\C}(\bar{Y})$ be the least $i$ and $j$
such that $\bar{X} \in \A_i$ and $\bar{Y} \in \C_j$, respectively
(note that this definition makes sense because $\C = \C_K$ and
$\A = \A_K$ by Proposition~\ref{prop:cor_alg}). Further, for
all $\alpha \in \bar{\Gamma}^*$ and $\bar{X} \in \bar{\Gamma}$ we define 
\begin{itemize}
\item $\price_{\C^*}(\alpha) = \begin{cases}
            \max\{0,I_{\C}(\alpha(i)) \mid 0 \leq i < \len{\alpha} \}
                                   & \mbox{if } \alpha \in \C^*\\
                            \infty & \mbox{otherwise.}
                             \end{cases}$
\item $\mathit{value}_{\C^*\A}(\alpha\bar{X}) = \begin{cases}
            \max\{\price_{\C^*}(\alpha), I_{\A}(\bar{X})\}
                                   & \mbox{if } \alpha \in \C^* 
                                     \mbox{ and } \bar{X} \in \A\\
                            \infty & \mbox{otherwise.}
                             \end{cases}$

\item $\price_{\C^*\A}(\alpha) = 
   \min\{\mathit{value}_{\C^*\A}(\beta) \mid \alpha=\beta\gamma \}$
\item $\price(\gamma) = \min\{\price_{\C^*\A}(\gamma),\price_{\C^*}(\gamma)\}$
\end{itemize}
Let $\sqsubset$ be a strict (i.e., irreflexive) ordering over 
$\C^* \A \bar{\Gamma}^* \cup \C^*$
defined as follows: $\alpha \sqsubset \beta$ if 
either $\price(\alpha) < \price(\beta)$, or 
$\price(\alpha) = \price(\beta)$ and 
$\price_{\C^*}(\gamma_1) < \price_{\C^*}(\gamma_2)$, where
$\alpha = \gamma_1\eta$, $\beta =\gamma_2\eta$, and 
$\eta$ is the longest common suffix of $\alpha$ and $\beta$.
One can easily verify that the ordering $\sqsubset$ is well-founded.
Let $\pi[\W_K],\pi[\U_K]$ be the SMD strategies of 
Lemma~\ref{lem-subset-prelim}. The strategy $\pi$ is defined so
that the following conditions are satisfied:
\begin{itemize}
\item if $\price(\bar{Z}\alpha) = \infty$, then $\pi(\bar{Z}\alpha)$
  is defined arbitrarily;
\item if $\bar{Z} \in \A$ and $I_{\A}(\bar{Z}) \leq \price(\bar{Z}\alpha)$,
  then $\pi(\bar{Z}\alpha) = \pi[\W_K](\bar{Z}\alpha)$;
\item otherwise, $\pi(\bar{Z}\alpha) = \pi[\U_K](\bar{Z}\alpha)$.
\end{itemize}
Observe that $\pi$ is regular, and the associated finite-state automaton
$\A_\pi$ is constructible in time polynomial in $\initdelta$ if the
strategies $\pi[\W_K],\pi[\U_K]$, the sets $\A,\C$, and the functions
$I_\A,I_\C$ have already been computed. Since all of these objects
are computable by a deterministic polynomial-time algorithm with 
$\NP\cap\coNP$ oracle,
the automaton $\A_\pi$ is also constructible by a deterministic
polynomial-time algorithm 
with $\NP\cap\coNP$ oracle. It remains to show that 
the definition of $\pi$ is correct, i.e., for every 
$\gamma \in \C^* \A \bar{\Gamma}^* \cup \C^*$ and every 
$\sigma \in \Sigma$ we have that 
$\calP^{\sigma,\pi}_{\gamma}(\R(T),\initdelta) < 1$. We proceed by  
induction with respect to the well-founded ordering $\sqsubset$.
The only minimal element of $\C^* \A \bar{\Gamma}^* \cup \C^*$ 
is $\varepsilon$ where we have 
$\calP^{\sigma,\pi}_{\varepsilon}(\R(T),\initdelta) = 0$. Now let 
$\bar{Z}\alpha \in \C^* \A \bar{\Gamma}^* \cup \C^*$ be some non-minimal
element. By Lemma~\ref{lem-subset-prelim} and the definition
of $\pi$ we immediately have that either
$\calP^{\sigma,\pi}_{\bar{Z}\alpha}(\R(T),G_{\initdelta}) < 1$
  or 
  $\calP^{\sigma,\pi}_{\bar{Z}\alpha}(\R(\gamma\alpha,G_{\initdelta})) > 0$
where $\gamma\alpha \sqsubset \bar{Z}\alpha$. In the first case,
we are done immediately, and in the second case we apply induction 
hypothesis.

\section{Conclusions}

We have solved the qualitative reachability problem for stochastic
BPA games, retaining the same upper complexity bounds that have
previously been established for termination \cite{EY:RMDP-efficient}. One 
interesting question which remains unsolved is the decidability
of the problem whether $\val(\alpha) = 1$ for a given BPA configuration
$\alpha$ (we can only decide whether player~$\Box$ has
a $({=}1)$-winning strategy, which is sufficient but not necessary
for $\val(\alpha) = 1$). Another open problem is
quantitative reachability for stochastic BPA games, where the 
methods presented in this paper seem~insufficient.

\bibliographystyle{plain}
\bibliography{str-long,concur}

\end{document}